\documentclass{vldb}
\usepackage{ntheorem}
\usepackage[f]{esvect}
\usepackage{subfig}
\usepackage{subfloat}
\usepackage{graphicx}
\graphicspath{{gnuplot//}}

\usepackage{yhmath}

\usepackage[ruled,vlined,linesnumbered]{algorithm2e}
\usepackage{algpseudocode}

\newtheorem{Definition}{Definition}

\newtheorem{Theorem}{Theorem}

\newenvironment{proofsketch}{\noindent{\sc Proof Sketch:} \hspace*{0.2em}}{\qed}
\theoremstyle{nonumberbreak}

\newcommand{\squishlist}
{
	\begin{list}{$\bullet$}
		{
			\setlength{\itemsep}{0pt}
			\setlength{\parsep}{3pt}
			\setlength{\topsep}{3pt}
			\setlength{\partopsep}{0pt}
			\setlength{\leftmargin}{1.5em}
			\setlength{\labelwidth}{1em}
			\setlength{\labelsep}{0.5em}
		}
	}
	
	\newcommand{\squishend}
	{
	\end{list}
}

\DeclareMathOperator*{\argmin}{\arg\!\min}

\makeatletter
\def\@copyrightspace{\relax}
\makeatother

\begin{document}

\title{Active Sampler: Light-weight Accelerator \\ for Complex Data Analytics at Scale}

\numberofauthors{1}
\author{
	\alignauthor
	\mbox{Jinyang Gao$^{\ddag}$ \qquad H. V. Jagadish$^{\S}$ \qquad Beng Chin Ooi$^{\ddag}$}\\
	\affaddr{$^\ddag$ National University of Singapore \qquad $^\S$ University of Michigan}\\
	\email{$^{\ddag}$ \{jinyang.gao, ooibc\}@comp.nus.edu.sg \qquad $^\S$ jag@umich.edu}
}

\maketitle

\hyphenation{sto-cha-stic}
\newcommand{\sqr}[1]{{(#1)}^2}
\newcommand{\wti}[0]{\vt{w_{t_i}}}
\newcommand{\pair}[1]{\langle #1 \rangle}
\newcommand{\norm}[1]{\uninorm{#1}^2}
\newcommand{\uninorm}[1]{{\Vert #1 \Vert}_2}
\newcommand{\lnorm}[1]{{\Vert #1 \Vert}_1}
\newcommand{\vt}[1]{{\bf #1}}
\newcommand{\fw}[0]{f_\vt{w}}
\newcommand{\fwt}[0]{f_\vt{w_t}}
\newcommand{\fwti}[0]{f_\wti}
\newcommand{\fwx}[0]{\fw(\vt{x})}
\newcommand{\fwxi}[0]{\fw(\vt{x_i})}
\newcommand{\fwtxi}[0]{\fwt(\vt{x_i})}
\newcommand{\fwtixi}[0]{\fwti(\vt{x_i})}
\newcommand{\datapair}[0]{\pair{\vt{x},y}}
\newcommand{\datapairi}[0]{\pair{\vt{x_i},y_i}}
\newcommand{\pxy}[0]{P_{\datapair}}
\newcommand{\lwxy}[0]{L(\fwx, y)}
\newcommand{\lwxiy}[0]{L(\fwxi, y)}
\newcommand{\lwxyi}[0]{L(\fwxi, y_i)}
\newcommand{\lwtxyi}[0]{L(\fwtxi, y_i)}
\newcommand{\lwtixyi}[0]{L(\fwtixi, y_i)}
\newcommand{\lwxyti}[0]{L(\fw(\vt{x_{t_i}}), y_{t_i})}
\newcommand{\rfw}[0]{R(\fw)}
\newcommand{\erfw}[0]{R_{emp}(\fw)}
\newcommand{\errfw}[0]{R_{reg-emp}(\fw)}
\newcommand{\dd}[0]{{~\rm d}}
\newcommand{\regw}[0]{\rho_f(\vt{w})}
\newcommand{\lossw}[0]{L(\vt{w})}
\newcommand{\losswt}[0]{L(\vt{w_t})}
\newcommand{\pii}[0]{{p_i}}

\newcommand{\giw}[0]{g_i({\vt{w}})}
\newcommand{\gt}[1]{g_t(#1)}
\newcommand{\gtw}[0]{\gt{\vt{w}}}
\newcommand{\gw}[0]{g({\vt{w}})}
\newcommand{\expeci}[1]{E_i[#1]}
\newcommand{\expec}[1]{E[#1]}
\newcommand{\pro}[1]{Pr(#1)}
\newcommand{\shortexp}[1]{e^{#1}}
\newcommand{\logistic}[1]{(1+\exp{(-#1)})}
\newcommand{\gradw}[0]{\nabla_{\vt{w}}}
\newcommand{\grad}[1]{\nabla_{#1}}

\newcommand{\pxyw}[0]{\pro{y | \fwx}}
\newcommand{\pxiyw}[0]{\pro{y | \fwxi}}
\newcommand{\uncertainwi}[0]{U(\vt{w}, \vt{x_i})}
\newcommand{\sigwi}[0]{S(\vt{w}, \vt{x_i})}
\newcommand{\viwi}[0]{I(\vt{w}, \vt{x_i})}
\newcommand{\pw}[0]{P_\vt{w}}
\newcommand{\pwx}[0]{\pw(\vt{x})}
\newcommand{\pwxi}[0]{\pw(\vt{x_i})}
\newcommand{\igwi}[0]{IG(\vt{w}, \vt{x_i})}
\newcommand{\igwj}[0]{IG(\vt{w}, \vt{x_j})}
\newcommand{\pxi}[0]{\pro{\vt{x_i}}}
\newcommand{\wxi}[0]{w_i}
\newcommand{\wxti}[0]{w_{t_i}}
\newcommand{\partialfr}[1]{\frac{\partial}{\partial #1}}
\newcommand{\gradwli}[0]{\gradw \lwxyi}
\newcommand{\gradwtili}[0]{\grad{\wti} \lwtixyi}
\newcommand{\normgradwli}[0]{\uninorm{\gradwli}}
\newcommand{\normgradwtili}[0]{\uninorm{\gradwtili}}
\newcommand{\gradfli}[0]{\partialfr{\fwxi} \lwxiy}
\newcommand{\normgradwfi}[0]{\uninorm{\gradw \fwxi}}
\newcommand{\normgradxi}[0]{\uninorm{\gradw \lwxiy}}
\newcommand{\expgradxi}[0]{E_y[\normgradxi]}
\newcommand{\var}[1]{Var{(#1)}}
\newcommand{\ems}[1]{{\em #1}}
\newcommand{\emb}[1]{{\bf #1}}
\newcommand{\eat}[1]{}

\begin{abstract}

Recent years have witnessed
amazing outcomes from ``Big Models'' trained by ``Big Data''.
Most popular algorithms for model training are iterative.
Due to the surging volumes of data,
we can usually afford to process only a fraction of the training data in each iteration.
Typically, the data are either uniformly sampled or sequentially accessed.

In this paper, we study how the data access pattern can affect model training.
We propose an \ems{Active Sampler} algorithm, where training data with more ``learning value'' to the model are sampled more frequently.
The goal is to focus training effort on valuable instances near the classification boundaries,
rather than evident cases, noisy data or outliers.
We show the correctness and optimality of Active Sampler in theory,
and then develop a light-weight vectorized implementation.
Active Sampler is orthogonal to most approaches optimizing the efficiency of large-scale data analytics, and can be applied to most analytics models trained by 
stochastic gradient descent (SGD) algorithm.
Extensive experimental evaluations
demonstrate that Active Sampler can speed up the training procedure of SVM,
feature selection and deep learning,
for comparable training quality %added by Jag, please make sure you agree
by 1.6-2.2x.
%, compared with uniform sampling.

\end{abstract}

\section{Introduction}

We live in an age of ever-increasing size and complexity of ``Big Data''. 
To understand the data and decipher the information that truly counts,
many advanced large-scale machine learning models have been devised,
from million-dimension linear models (e.g. Logistic Regression~\cite{logistic}, Support Vector Machine~\cite{pegasos}, feature selection~\cite{largelasso,systemfs}, Principal Component Analysis~\cite{sgdpca}) to complex models like Deep Neural Networks~\cite{systemdl2} or topic models~\cite{svi}. 
While these models have demonstrated value for a wide spectrum of applications~\cite{url,largelasso,cnn}, their complexity causes the training cost to increase dramatically with the surging volume of data. 
This difficulty with scale severely affects the viability of many advanced models on industry-scale applications.
%This however causes the training cost to increase dramatically,
% severely affecting the viability of such complex data analytics. 
Consequently, accelerating the training procedure of those ``Big Models'' on ``Big Data'' 
has attracted a great deal of interest.

\begin{figure}
	\centering
	\subfloat[Easy]{
		\epsfig{height=3.2cm,width=3.2cm,file=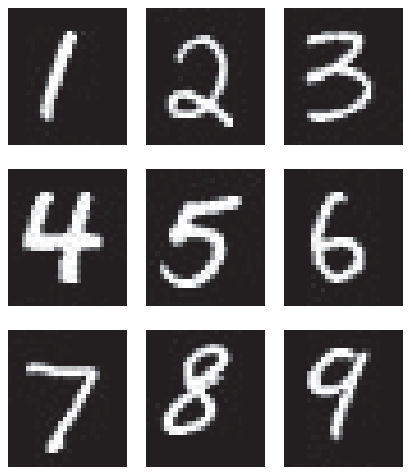}}
	\hspace{6mm}\subfloat[Hard]{
		\epsfig{height=3.2cm,width=3.2cm,file=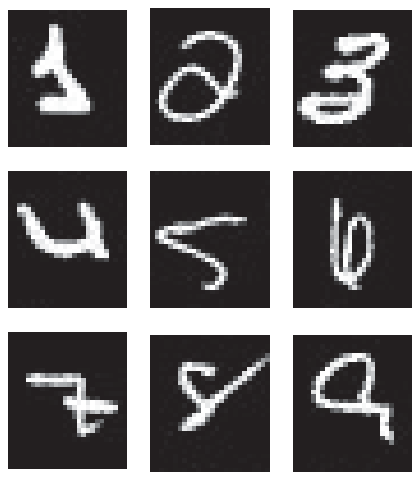}}
	\caption{Information Gain from Training Data}
	\label{fig:intuition}
	%\vspace{-0.3cm}
\end{figure}

% reduce the training cost -- SGD
 All the models mentioned in the preceding paragraph, and many others, can be formulated as
minimizing a specific objective function based on a set of data observations,
i.e. the Empirical Risk Minimization~\cite{statisticalml} (ERM) problem.
Even though gradient descent~\cite{boydcvx,gdsurvey} has been widely used
for decades, evaluating the full gradient over all the training samples
(i.e. batch gradient descent~\cite{batchgd}) is extremely expensive
in the prevailing scale of millions of training samples.
%In some cases, the full gradient over all training samples can be evaluated
%via some subtle techniques~\cite{one-parameter exponential families}.
%However, in other cases, batch gradient descent requires
%evaluating the gradient of objective function based on each training sample.
%In typical cases,
%batch gradient descent requires
%evaluating the gradient of objective function based on each training sample.
To reduce the computational cost at every iteration,
Stochastic Gradient Descent (SGD)~\cite{sgd,largesgd} optimizes the objective function
based on a single random sample at each iteration.
Thereby, the computation cost per iteration is reduced greatly, but now 
many more iterations are required to reach a certain degree of accuracy or finally converge~\cite{sgdconvergence}.
This is because the stochastic gradient used
in each iteration is highly sensitive to the specific random sample chosen.
Although the expectation of stochastic gradient is exactly the full gradient,
the large variance causes the direction of stochastic gradient to deviate from
that of the full gradient,
which is the optimal direction to minimize the objective function.
Some samples may even direct the model to the opposite of the correct direction.

It has been shown that reducing the variance of stochastic gradient~\cite{sgdvr,svrg,sag} will lead to a much faster convergence rate.
A commonly used scheme called Mini-batch SGD~\cite{minibatch} is developed for this purpose: by averaging the gradient from a mini-batch of samples, the variance of gradient is significantly reduced, at the cost of some increased computation per iteration.  
The optimal mini-batch size is determined as a trade-off between the increased computational cost per iteration for a larger mini-batch and the increased variance (and hence number of iterations) for a smaller mini-batch.  

% active sampler -- idea

In this paper, we seek to further optimize the stochastic gradient
by not merely averaging gradients from more random samples but rather improving
the quality of data samples.
To this end, we propose a light-weight SGD accelerator
inspired by active learning~\cite{activesurvey1,activesurvey2}.
In active learning, training data are selected to maximize the ``learning value''. 
For example, to train a classification model, training data points near
class boundaries are more valuable than points in the interior of a class. 
We adapt the idea of active learning to the SGD optimization context by
choosing samples from not a uniform distribution over the training data but
rather a biased distribution from which we expect to learn more.

% active sampler -- example

%\begin{figure}
%	\centering
%	\epsfig{width=8cm,file=figure/intro.eps}
%	\caption{Information Gain from Training Data}
%	\label{fig:intuition}
%\end{figure}

Figure~\ref{fig:intuition} gives an example of
how different samples can affect the training efficiency.
The left side contains some images of written digits that are very easy to classify.
For these ``easy-to-classify'' images, most models can classify them correctly
even after a handful of steps.
In subsequent thousands of iterations or even more,
these easy-to-classify images will be sampled and correctly classified with
high confidence,
contributing almost zero gradient to the model.
Consequently, the training time consumed by those easy-to-classify images is largely wasted.
In contrast,
the right side of Figure~\ref{fig:intuition} contains some images that are hard to classify.
By putting more effort on those ``hard-to-classify'' images,
the accuracy of model may improve at a faster rate.

% active sampler -- theory result

Based on the above intuition, we develop a weighted sampling method called \ems{Active Sampler}.
We find that to maximize the information gain in each iteration, the sample frequency for each training sample should be proportional to the estimate of its gradient magnitude.
%Theoretical analyses are also given to demonstrate the correctness and optimality of active sampler.
Notwithstanding the sampler itself is biased, we show that the original objective function based on uniform weight
can still be correctly minimized by re-weighting the gradient of each sample.
From the view of variance reduction, Active Sampler also provides ta gradient with the smallest variance compared with all weighted sampling methods, including the uniform random sampling.
The net result is a system that requires far fewer iterations for model convergence (or to reach a required accuracy threshold).

% practical considerations -- optimazation will lead to overhead

Although optimization methods can reduce the number of iterations,
it should also be noted that they may introduce additional computation cost in each iteration.
There have been a great amount of research works~\cite{sag, svrg,s3gd} focusing on accelerating SGD.
Theoretically, these methods have significantly faster convergence rate.
While Momentum and AdaGrad~\cite{adagrad} methods have shown
their effectiveness and have been integrated into practical SGD solver,
most methods are far from being used in practice due to their significant additional computation cost.  For example, the cost of SVRG~\cite{svrg} per iteration is at least three times the cost of standard SGD per iteration. As noted in~\cite{s3gd}, mini-batch SGD still dominates in most cases, due to its light-weight computation and good vectorization.

% a light weight scheme for active sampler

To make Active Sampler as efficient as possible in practice,
the design principle is to reduce the overhead involved in each iteration.
In the actual implementation, we use lots of existing knowledge to approximate the information that needs additional computation cost, and the sampling distribution is decided by the gradient magnitude in previous iterations. To evaluating the gradient magnitude for each sample, an effective scheme is applied to avoid the explicit calculation of the gradient of each sample. This scheme makes it possible that the computation for multiple samples can still be efficiently vectorized.  Moreover, the computation cost is only $O(m+l)$ for a $m \times l$ parameter matrix. Therefore, for Active Sampler, the computation overhead introduced in each iteration is %extremely
 light-weight, considering its significant contribution in reducing the number of total iterations.

% contributions

The contributions of our work are summarized as follows:

\squishlist

\item We propose a general SGD accelerator,
called Active Sampler, where more informative training data
is sampled more frequently for model training.
We formalize the problem as SGD optimization for ERM with weighted sampling, and show that the Active Sampler has the largest information gain and the smallest variance among all weighted sampling solutions.

\item We develop a light-weight and fully vectorized algorithm for Active Sampler, making the computation cost of Active Sampler in each iteration comparable to the naive mini-batch SGD.

\item We implement the Active Sampler framework and evaluate its
performance on three popular machine learning algorithms: SVM, feature selection and deep neural network.
Active Sampler reduces the number of iterations to
reach a certain accuracy by half,
while only consuming 10\%-20\% additional computation cost in each iteration.
In short, Active Sampler speeds up the training procedure by more than 1.6x.

\squishend

The remainder of this paper is organized as follows. We first introduce the background in Section~\ref{sec:prel}.
Then we propose the Active Sampler, show its effectiveness in theory and discuss its practical implementation issues in Section~\ref{sec:alg}.
The experimental results are discussed in Section~\ref{sec:exp}. Finally, we review about the related works in Section~\ref{sec:related} and conclude at Section~\ref{sec:conclusion}.

\section{Preliminaries}
\label{sec:prel}

In this section, we first introduce the Empirical Risk Minimization (ERM) framework,
and show its connection to data analytics models.
Then we describe the stochastic gradient descent algorithm, a general solver for the ERM problem,
which we aim to improve upon in this work.

\subsection{Empirical Risk Minimization}

\begin{table}[t]\caption{Examples of ERM Applications}
	\label{tbl:erm}
	\vspace{-0.3cm}
	\center{
		\tiny
		\begin{tabular}{|c|c|c|c|}
			\hline
			Algorithm & $\fwx$ & $\lwxy$ & $\regw$\\
			\hline\hline
			Linear Regression & $\vt{w^T}\vt{x}$ & {\tiny ${(y - \fwx)}^2$} & 0 or $\lambda \norm{\vt{w}}$ \\
			\hline
			Hinge-loss SVM & $\vt{w^T}\vt{x}$ & {\tiny $\max(0, 1 -\fwx*y)$} & 0 or $\lambda \norm{\vt{w}}$\\
			\hline
			Logistic Regression & $\vt{w^T}\vt{x}$ & {\tiny $\log\logistic{\fwx*y}$ } & 0 or $\lambda \norm{\vt{w}}$\\
			\hline
			Feature Selection & $\vt{w^T}\vt{x}$ & {\tiny $\log\logistic{\fwx*y}$ } & $\lambda \lnorm{\vt{w}}$\\
			\hline
			Neural Network & complex & {\tiny $\log\logistic{\fwx*y}$ } & 0 or $\lambda \norm{\vt{w}}$\\
			\hline
			PCA & {\tiny $\vt{w}\vt{w}^T\vt{x}$ - $\vt{x}$} & {\tiny $\norm{\fwx}$ } & 0 \\
			\hline
		\end{tabular}
	}
	%\vspace{-0.2cm}
\end{table}

Empirical Risk Minimization (ERM) is a principle in the statistical learning theory which
forms the basis for defining a family of analytics models.
From the view of ERM, the central idea in machine learning is to learn a model and use it to approximate the data.
The difference between the approximation and the real data is then measured by a loss function, which should be minimized by tuning the parameters of the model.
For the sake of simplicity, in this work we formalize ERM from the supervised learning perspective, where each training instance is a pair $\datapair$ consisting of content $\vt{x}$ and label $y$. For unsupervised problems, label $y$ is a null term $\Phi $, and data can be represented as $\pair{\vt{x},\Phi}$.

\begin{Definition}[Loss Function]\label{def:loss}
	Given a data instance represented as $\datapair$, and a model hypothesis $\fw$ (i.e. a model $f$ with parameter $\vt{w}$), the loss function $\lwxy$ is a disagreement measure function between the model approximation (i.e. prediction) $\fwx$ and the actual label $y$.
\end{Definition}

Here we shall give several examples on the loss measures that are commonly used.
For hard classification problems, the loss measure function can be written as: $\lwxy = I[\fwx \neq y]$.
For linear regression, the loss measure function can be written as: $\lwxy = \norm{\fwx-y}$.
For general linear models where $\fwx = \vt{w}^T\vt{x}$ and $y = \pm 1$, if the loss measure is $\lwxy = \max(0, 1 -\fwx*y)$, then the model is a hinge-loss SVM; if the loss measure is the log logistic function $\lwxy = \log \logistic{\fwx*y}$, then the model is a logistic regression.
For PCA, by using $\vt{w}$ to denote the low rank projection matrix, the problem can be formulated
as $\fwx = \vt{w}\vt{w}^T\vt{x} - \vt{x}$, and $\lwxy = \norm{\fwx}$.

After defining the measure of disagreement between model output and target label, the ultimate goal of learning is naturally to minimize the total disagreement by tuning the model parameters.
This is called \ems{Risk Minimization}~\cite{statisticalml}, which is defined as follows:

\begin{Definition}[Risk Minimization]\label{def:risk}
	Let $\pxy$ to be the distribution of data, the risk associated with model hypothesis $\fw$ is defined as the expectation of the loss for the potential data distribution:
	\begin{equation}
		\rfw = \expec{\lwxy} = \int{\lwxy \dd\pxy}
	\end{equation}
	The goal of learning algorithms is to find the parameter $\vt{w}$ that minimizes the risk:
	\begin{equation}
		\argmin_{\vt{w}}{\rfw}
	\end{equation}
\end{Definition}

However, in general, $\rfw$ cannot be directly minimized since the exact latent data distribution $\pxy$ is unknown. Instead, the common way is to use the distribution of training data to approximate $\pxy$.
Therefore, the \ems{Empirical Risk}~\cite{statisticalml} is used as the optimization target.

\begin{Definition}[Empirical Risk]\label{def:erm}
	The empirical risk is defined as the average of loss on the training set with $n$ instances.	
	\begin{equation}
		\erfw =  \frac{1}{n}\sum_i\lwxyi
	\end{equation}
	To simplify the notation, we use $\lossw$ to denote $\erfw$.
\end{Definition}

According to the VC-dimension theory~\cite{statisticalml}, the difference between real risk and empirical risk may be large when the model hypothesis $\fw$ is too complex while the size of training data $n$ is not large enough. This phenomenon is called over-fitting. To prevent over-fitting, the empirical risk is often regularized to penalize the complexity of model $\fw$:

\begin{Definition}[ERM with Regularization]\label{def:rerm}
	The empirical risk with regularization is defined as the average of loss function on the training set, plus a penalty regularization term $\regw$ based on the complexity of the model $\fw$.
	\begin{equation}
		\errfw =  \lossw + \regw
	\end{equation}
	In ERM with regularization, the goal of learning algorithm is to minimize the empirical risk with regularization, i.e.,
	\begin{equation}
		\label{eq:erm}
		\argmin_{\vt{w}}{\errfw}
	\end{equation}
\end{Definition}

For most applications, we use the $l_2$-norm of parameter $\lambda \norm{\vt{w}}$ as the regularization function $\regw$.
This is actually a Gaussian prior over the parameter distribution from the Bayesian view.
For feature selection methods such as Lasso~\cite{largelasso}, 
the $l_1$-norm regularization $\lambda \lnorm{\vt{w}}$ is used to select those sparse features.

We list some examples of analytics models from ERM family in Table~\ref{tbl:erm},
and show their connections.
\eat{When we use $\vt{w}$ as the variable, we assume that
$\lwxyi$ and $\regw$ are $\gamma$-Lipschitz continuous for some constant $\gamma$ in most cases (i.e. we can get the gradient of these functions, and the gradient is bounded by some constant).
We can see that all the above models follow this assumption.}

\subsection{Stochastic Gradient Descent}

\begin{table}[t]\caption{Common Notations}
	\label{tbl:notations}
	\vspace{-0.3cm}
	\center{
		\begin{tabular}{|c|l|}
			\hline
			Notation & Meaning\\
			\hline\hline
			$\datapair$ & training instance\\
			\hline
			$\vt{w}$ & model parameters\\
			\hline
			$\fwx$ & model prediction for data $\vt{x}$\\
			\hline
			$\lwxy$ & loss on instance $\datapair$\\
			\hline
			$\lossw$ & empirical risk\\
			\hline
			$\regw$ & regularization term\\
			\hline
			$\gradw$ & gradient operator\\
			\hline
			$\gradw \lossw$ & batch gradient\\
			\hline
			$\giw$ & stochastic gradient\\
			\hline
			$\pii$ & sampling probability for $\datapairi$ \\
			\hline
			$\var{\giw}$ & scalar variance of $\giw$ \\
			\hline
		\end{tabular}
	}
	%\vspace{-0.2cm}
\end{table}

To optimize the ERM problem described in Equation~\ref{eq:erm},
batch gradient descent method is used to
iteratively alter the parameter towards the fastest direction
to minimize the objective function.
By defining the step size using the learning rate $\eta$,
batch gradient descent method uses the following updating rule to optimize the parameter:
\begin{eqnarray}
	\label{eq:gd}
%	\nonumber \vt{w}_{new} & = & \vt{w} - \eta \gradw {\errfw} \\
%	\nonumber & = & \vt{w} - \eta \gradw {(\lossw + \regw)} \\
    \nonumber \vt{w}_{new} & = & \vt{w} - \eta \gradw {(\lossw + \regw)} \\
	& = & \vt{w} - \eta \gradw \regw - \frac{\eta}{n}\sum_i{\gradw \lwxyi}
\end{eqnarray}
%As we can observed from Equation~\ref{eq:gd}, at each step we need to compute the gradient of $\vt{w}$ based on each training sample $\datapair$, i.e. $\gradw \lwxyi$,  where the cost is very expensive.
As can be observed from Equation~\ref{eq:gd},
we need to evaluate the gradients $\gradw \lwxyi$ for all training instances at each step,
making the computation cost of $\gradw \lossw$ extremely expensive.
To avoid this cost,
stochastic gradient methods use an inexact gradient which is estimated from random samples.

\begin{Definition}[Stochastic Gradient Descent]\label{def:sgd}
	In stochastic gradient descent,
	the true gradient $\gradw \lossw$ is approximated by a stochastic gradient $\giw$.
	\begin{eqnarray}
		\label{eq:sgd}
		\vt{w}_{new} & = & \vt{w} - \eta \gradw \regw - \eta \giw
	\end{eqnarray}
	Taking $\giw$ as a random variable, the expectation of $\giw$ should equal to the gradient of $\lossw$, i.e.
	\begin{eqnarray}
		\expeci{\giw} & = & \gradw \lossw
	\end{eqnarray}
	In the standard SGD algorithm, $\giw$ is obtained by simply evaluating the gradient at a random single instance $i$:
	\begin{eqnarray}
		\giw & = & \gradw \lwxyi
	\end{eqnarray}
	where $i$ is randomly drawn from $\{ 1, ..., n\}$, and the sampling probability $\pii$ for each instance $i$ is $1/n$.
	The scalar variance of stochastic gradient is denoted as $\var{\giw}$, defined by $\expeci{\norm{\giw - \gradw \lossw}}$, which is a scalar instead of the covariance matrix.
\end{Definition}

Table~\ref{tbl:notations} lists most of the important notations used throughout this paper.
Unless otherwise specified, variance used in the paper refers to the scalar variance.
Intuitively, the requirement $\expeci{\giw} = \gradw \lossw$ is to guarantee that the SGD algorithms will converge at the optimal point~\cite{sgdconvergence}, as the expectation of update in SGD will be a zero vector at the point where batch gradient descent algorithm converges.
Obviously, the standard SGD algorithm satisfies this requirement.
\begin{eqnarray}
	\nonumber \expeci{\giw} & = & \sum_i \pii \giw \\
	\nonumber & = & \sum_i \frac{1}{n} \gradw \lwxyi \\
	\nonumber & = &  \gradw \frac{1}{n} \sum_i \lwxyi \\
	& = &  \gradw \lossw
\end{eqnarray}

Although the training cost per iteration for SGD is extremely light-weight compared to the batch gradient algorithm,
the main drawback of SGD is that $\giw$ is not the exact $\gradw \lossw$.
Therefore, the direction of the stochastic gradient $\giw$ differs from the optimal direction $\gradw \lossw$. This phenomenon causes SGD algorithm to be less efficient and take more iterations to converge or reach a certain accuracy.
It has been shown by~\cite{sgdvr,svrg,sag} that simply reducing the variance of stochastic gradient will increase the convergence rate of SGD algorithms.
However, most variance reduction techniques require additional computation
cost compared with the standard SGD.
In essence, 
there is a trade-off between the number of iterations
required to reach a certain accuracy and the computational cost per iteration.
Therefore, the main goal of optimizing SGD algorithm is to reduce the variance of $\giw$,
while keeping the computation of $\giw$ light-weight.

\section{Active Sampler}
\label{sec:alg}

\subsection{Overview}

In this paper, we revisit the SGD algorithm from a brand new angle -- the information gain of the
model at each iteration.
We can regard the SGD algorithm as an active learning procedure which sequentially takes
samples from the dataset and refines its model.
Naturally,
training the model using samples with more information
would facilitate faster improvement of the model.
We term our weighted sampling strategy as \ems{Active Sampler}.
The intuition here is that a larger number of training samples
are not helpful for refining the model
(or at least not helpful at a certain training stage), 
including data that are too evident to predict,
data that are noisy, and
data that have just been visited.
By simply skipping these samples,
we can save a significant amount of training time.
In contrast,
the samples that are close to
the border of class may be very helpful to refine the model (or even define the model in some cases such as SVM).
This idea is very similar to active learning but with a major difference
-- the objective of active learning is to reduce the number of training samples,
while the objective of active sampling is to reduce the number of training iterations.

To exploit information gain as a means to speed up SGD training, three issues
have to be addressed:
(1) define what is the information gain for model training from each training sample;
(2) adapt the Active Sampler into the SGD framework and study how information gain
can help speeding up SGD;
(3) design a light-weight implementation that can be applied to real systems.
We shall address these three issues in the following subsections.

\subsection{Information Gain}

In this subsection,
we define the information gain directly following the basic intuition from the training of typical soft margin classifiers.
Later, we will fit this initial idea
into a formal SGD framework and provide
a rigorous theoretical analysis that illustrates how this strategy can benefit the
performance of SGD for all sorts of ERM problems in the next subsection.

In soft margin classifiers,
instead of giving a single label as the prediction,
the classifier outputs a probabilistic distribution over latent labels.
Recently proposed classification algorithms (e.g. Lasso, Neural Network and Soft SVM etc.)
are typically trained as soft margin classifiers,
since the continuous optimization is much efficient than the discrete optimization, which is a NP-Hard problem.
%%% ooibc: pls check

%\begin{Definition}[Soft Margin Classifer]\label{def:smc}
%	Given a predictor $\fw$, its related soft probabilistic distribution $\pw$ is generated via logistic function:
%	using $\pxyw$ to denote $\pro{y | \pwx}$, $\pxyw = 1 / \logistic{\fwx * y}$.
%	The objective function of soft margin classifiers are usually the log-likelihood function and can be easily optimized via gradient methods, i.e. the loss function is $\lwxy = - \log \pxyw = \log\logistic{\fwx*y}$.
%	Obviously, all algorithms in Table~\ref{tbl:erm} using logistic loss function are soft margin classifiers.
%\end{Definition}

\begin{Definition}[Soft Margin Classifier]\label{def:smc}
	Given a predictor $\fw$, the soft probabilistic classifier is defined by using log logistic function as the loss function, i.e.
	\begin{eqnarray}
		\label{eq:smc}
		\lwxy = \log\logistic{\fwx*y}
	\end{eqnarray}
	Obviously, all algorithms in Table~\ref{tbl:erm} using the logistic loss function are soft margin classifiers.
	The rationale is that the logistic function is used to transfer the prediction $\fw \in R$ to a classification probability, i.e.
	\begin{eqnarray}
		\label{eq:smcpro}
		\pxyw = 1 / \logistic{\fwx * y}
	\end{eqnarray}
	Then the log-likelihood $\log \pxyw$ is maximized, i.e. the loss is $\log\logistic{\fwx*y}$.
\end{Definition}

\begin{figure}[t]
	\centering
	\epsfig{width=6cm,file=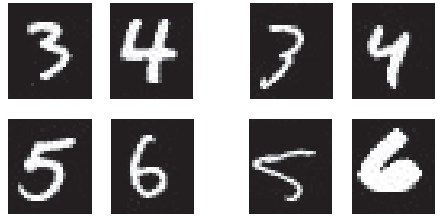}
	\caption{Uncertainty}
	\label{fig:uncertain}
	%\vspace{-0.3cm}
\end{figure}

Now, let us analyze the possible factors that may affect the information
gain of a model from each training sample.
First, from the view of active learning,
the information gain by revealing a label which can be easily predicted is very limited.
This is also true for the SGD algorithm -- visiting a sample that the model can always classify correctly is not helpful,
as the loss cannot be further reduced if the loss is already close to $0$.
A model can only learn from samples which are uncertain in prediction.
Figure~\ref{fig:uncertain} shows an example to illustrate why uncertainty helps the model to improve:
the images on the left
are very easy to predict, and therefore $\pxyw$ is almost 100\%.
As a result, the loss $\lwxy$ for each of those images is almost zero and has little room for further optimization.
In contrast,
the images on the right is much harder to predict.
By selecting them as samples for optimization purpose,
the loss $\lwxy$ on those samples could be significantly reduced and the average performance of model is hence improved. In information theory~\cite{statisticalml},
the information gain by revealing a random variable is usually defined as the entropy of that random variable:
     
\begin{Definition}[Uncertainty]\label{def:uncertain}
	The uncertainty of a training instance $\vt{x_i}$ for a model $\fw$ is defined as the entropy of the $\fwxi$, i.e.
	\begin{eqnarray}
		\label{eq:uncertain}
		\uncertainwi = - \sum_y \pxiyw \log \pxiyw
	\end{eqnarray}
\end{Definition}

Second, not all the training instances contribute equally
to the model performance.
For example, though the model may always be uncertain about the label for noisy data,
this does not mean that noisy data are more helpful to improve the model performance.
This is because the uncertainty measure only evaluates the information inside the label,
but not its contribution to the model.
Therefore, we introduce another measure called \ems{significance}
to evaluate the efficiency of information transfer from the data to the model.
Intuitively, the output of a noisy instance is not sensitive to the change of the parameter.
When the output of a data instance is sensitive to the change of parameter, its loss will be significantly reduced even with tiny changes of the parameter,
which provides a clear instruction on how to reduce the loss by tuning the parameter.
The right hand side of Figure~\ref{fig:sig} shows the images that are noisy and with less significance.

\begin{Definition}[Significance]\label{def:sig}
	The significance of a training instance $\vt{x_i}$ for a model $\fw$ is defined as its sensitivity to parameter change:
	\begin{eqnarray}
	\label{eq:sig}
	\sigwi = \uninorm{\gradw \fwxi}
	\end{eqnarray}
	$\uninorm{\gradw \fwxi}$ is the maximal change of $\fwxi$ when the parameter $\vt{w}$ changes an unit distance.
\end{Definition}

\begin{figure}[t]
	\centering
	\epsfig{width=6cm,file=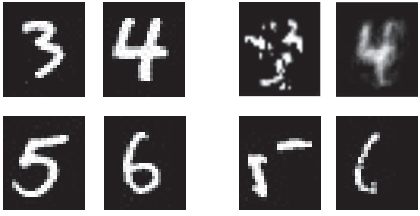}
	\caption{Significance}
	\label{fig:sig}
	%\vspace{-0.3cm}
\end{figure}

Here, we give a comparison between uncertainty and significance:
uncertainty measures the expectation of accuracy on the current model, while significance measures the potential improvement of accuracy by tuning the model.
Therefore, 
instances that are easy to classify usually have low uncertainty but high significance;
noisy instances usually have high uncertainty but low significance,
while valuable border instances that have not been well learned usually have both high uncertainty and high significance.

Third, the information gain in one iteration may overlap with the information obtained in
earlier training steps.
For example, visiting the training instance that has just been
trained usually does not provide extra information than what has been
derived in the previous visit.
\eat{The information of these two steps provided overlap with each other,
	since the change of model within one iteration is trivial.}
However, for a completely new instance that hasn't been trained, there may be no overlap between the information gain and information in previous steps.
Therefore, we use the visiting interval to measure the effect of information overlap:

\begin{Definition}[Interval]\label{def:his}
	The visiting interval $\viwi$ of a training instance $\vt{x_i}$ for a model $\fw$ is defined as the number of iterations since the last time $\vt{x_i}$ was used in training.  A larger interval provides less information overlap and more pure information gain.
\end{Definition}

Combining all the three factors together,
we define the information gain of model $\fw$ from training instance $\vt{x_i}$ as $\igwi$:
\begin{eqnarray}
\label{eq:igwi}
\igwi = \uncertainwi * \sigwi * \viwi
\end{eqnarray}
The objective of our Active Sampler is to choose the training instance $\vt{x_i}$ with the
largest $\igwi$.

\begin{Theorem}[Information Gain Maximization]
	\label{thm:intuition}
	By ch\-oosing the largest $\igwi$ in each iteration, the sampling frequency $\pii$ of each training instance $\vt{x_i}$ should be proportional to its expectation of the gradient magnitude, i.e.
	\begin{eqnarray}
	\label{eq:theintuition}
	\pii = \frac{\expgradxi}{\sum_i \expgradxi}
	\end{eqnarray}
\end{Theorem}
\begin{proofsketch}
At each iteration, $\igwi$ for each instance will increase $\uncertainwi$ * $\sigwi$, and the instance with the largest $\igwi$ will be selected as sample. After an instance is sampled, its $\igwi$ will be set to zero as its $\viwi$ becomes zero. Therefore, as the number of iterations grows, the sampling frequency for each instance should be proportional to its $\uncertainwi$ * $\sigwi$, considering that the largest $\igwi$ selected in each iteration should has a similar value. 	
%	As $\viwi$ will increase $1$ at each iteration, $\igwi$ for each instance will increase $\uncertainwi$ * $\sigwi$ or set to zero if it is selected.
%	At each iteration, the change to $\sum_i \igwi$ is $\sum_i \uncertainwi * \sigwi - \igwi$.
%	However, at any iteration, $\sum_i \igwi$ is between $0$ and $n\igwi$ since $0 \leq \igwj \leq \igwi$ if $i \neq j$. 	 
%	Therefore, the average of the largest $\igwi$ in each iteration 
%	converges to $\sum_i \uncertainwi * \sigwi$ as the number of iterations grows.	
%	Concomitantly, for each instance,
%	$\pii \igwi$ converges to to $\uncertainwi * \sigwi$ as the number of iterations grows.
%	Therefore the selected frequency $\pii$ should be proportional to $\uncertainwi * \sigwi$. 
	Meanwhile,
	\begin{eqnarray}
	\nonumber    & \expgradxi \\
	\nonumber   =&  \sum_y \pxiyw \normgradxi \\
	\nonumber   =& \sum_y \pxiyw \gradfli \normgradwfi  \\
	\nonumber   =& \sum_y (\pxiyw * \log \pxiyw) \normgradwfi  \\
	=& \uncertainwi * \sigwi
	\end{eqnarray}
\end{proofsketch}

\subsection{Weighted SGD Algorithm and Analysis}

While the above intuition suggests that different training instances 
should be sampled at different frequencies,
directly changing the sampling frequency will result in a bias in
the target of an optimization.
\begin{eqnarray}
	\label{eq:unequalweight}
	\nonumber \expeci{\giw} & = & \sum_i \pii \giw \\
	& = & \gradw \sum_i \pii \lwxyi
\end{eqnarray}
The loss function to be minimized is $\sum_i \pii \lwxyi$ instead of $\sum_i \frac{1}{n} \lwxyi$.
The weight for each training instance is unequal, which may affect the accuracy of the model.
For this reason, 
using $\gradw \lwxyi$ directly as $\giw$ is inappropriate. 
Instead, we should guarantee that $\expeci{\giw}$ is $\gradw \lossw$.

\begin{Theorem}[Weighted SGD]
	Given any sampling distribution $\{p_1$, ..., $\pii$, ..., $p_n\}$, 
	to get a SGD algorithm that optimizes $\lossw$, $\giw$ should be re-weighted to $\frac{\gradw \lwxyi}{n * \pii}$.
%%%ooibc: please check; not so smooth
\end{Theorem}
\begin{proof}
	To get a SGD algorithm that optimizes $\lossw$, 
    we need $\expeci{\giw}$ to be $\gradw \lossw$.
	By scaling $\giw$ $\wxi$ times and solve $\expeci{\giw} = \gradw \lossw$, we have:
	\begin{eqnarray}
		\label{eq:equalweight}
		\nonumber  \expeci{\giw} & = & \gradw \lossw\\
		\nonumber \Rightarrow \quad \quad \expeci{\wxi \gradw \lwxyi} & = & \gradw \lossw\\
		%\nonumber \Rightarrow  \sum_i \pii \wxi \gradw \lwxyi & = & \gradw \lossw\\
		\nonumber \Rightarrow  \quad \sum_i \pii \wxi \gradw \lwxyi & = &  \sum_i \frac{1}{n} \gradw \lwxyi\\
		\Rightarrow \quad \quad \quad \quad \quad \quad \quad \quad \quad \quad \pii \wxi  & = & \frac{1}{n}
	\end{eqnarray}
	Therefore, $\wxi = 1 / (n * \pii)$.
\end{proof}

Next, we show that setting $\pii$ proportional to the gradient magnitude will minimize 
the variance in stochastic gradient $\giw$ in all weighted sampling solutions 
that $\expeci{\giw} = \gradw \lossw$.

\begin{figure}[t]
	\centering
	\epsfig{width=8.5cm,file=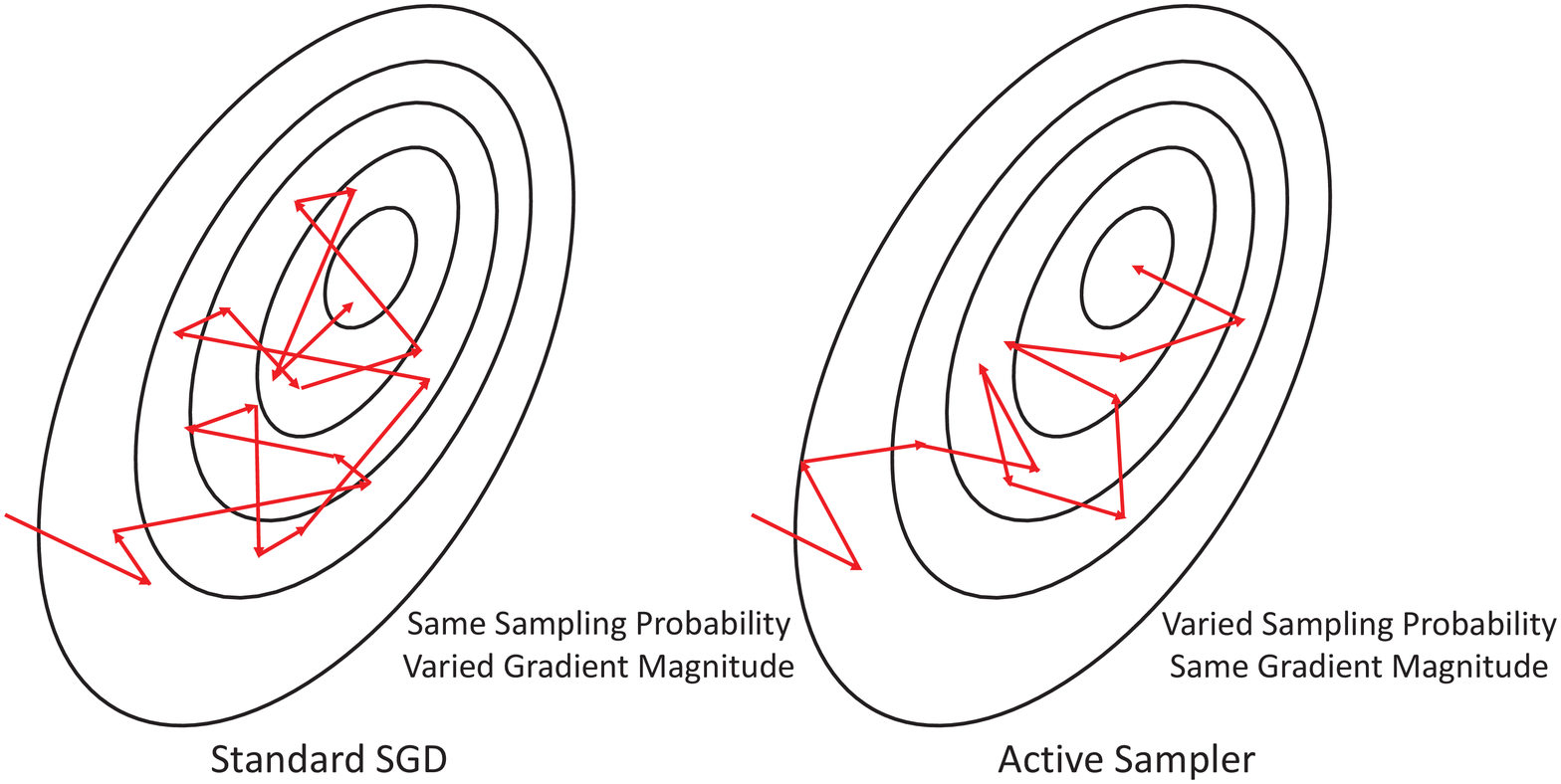}
	\caption{Comparison with standard SGD}
	\label{fig:sgd}
	%\vspace{-0.3cm}
\end{figure}

\begin{Theorem}[Optimal Weighted SGD]
	\label{thm:main}
 Let$\pii$ denote the sampling probability of training instance $\datapairi$ in a weighted SGD algorithm 
 where $\expeci{\giw} = \gradw \lossw$.
$\giw$ becomes a stochastic gradient with value $\frac{\gradw \lwxyi}{n * \pii}$ with sampling probability $\pii$.
To get the weighted SGD algorithm with the smallest variance of stochastic gradient (i.e. $\var{\giw}$), 
for each instance, $\pii$ should be proportional to its magnitude of stochastic gradient $\normgradwli$.
%%% ooibc: i edited, but it is still not clear
%%%    need to reword
	\begin{eqnarray}
		\pii = \frac{\normgradwli}{\sum_i \normgradwli}
	\end{eqnarray}
\end{Theorem}
\begin{proof}	
	\begin{eqnarray}
		\nonumber & \var{\giw} \\
		\nonumber = & \expeci{\norm{\giw}} - \norm{\expeci{\giw}} \\
		\nonumber = & \expeci{\frac{\norm{\gradwli}}{{(n\pii)}^2}}  - \norm{\expeci{\frac{\gradwli}{n\pii}}}\\
		\nonumber = & \sum_i \pii \frac{\norm{\gradwli}}{{(n\pii)}^2} - \norm{\sum_i \pii \frac{\gradwli}{n\pii}}\\
		= & \sum_i \frac{\norm{\gradwli}}{n^2\pii} - \norm{\gradw \lossw}
	\end{eqnarray}
	To minimize $\var{\giw}$ by tuning $\pii$, subjecting to $\sum_i \pii = 1$, according to Lagrange multiplier method, we have:
	\begin{eqnarray}
		\left\{
		\begin{array}{c}
			\label{eq:first}
			\partialfr{\pii} (\var{\giw}+ \lambda (\sum_i \pii -1)) = 0 \\
			\label{eq:second}
			\partialfr{\lambda} (\var{\giw}+ \lambda (\sum_i \pii -1)) = 0 \\
		\end{array}
		\right.
	\end{eqnarray}
	by solving Equation~\ref{eq:first}, we have:
	\begin{eqnarray}
		\lambda - \frac{\norm{\gradwli}}{{(n\pii)}^2} = 0
	\end{eqnarray}
	%and by solving Equation~\ref{eq:second}, we get:
	%\begin{eqnarray}
	%\sum_i \pii = 1
	%\end{eqnarray}
	Therefore, $\frac{\norm{\gradwli}}{{(n\pii)}^2}$ is a constant value for all training instances, which means $\normgradwli$ is proportional to $\pii$.
	Considering that $\sum_i \pii = 1$, we have:
	\begin{eqnarray}
		\pii = \frac{\normgradwli}{\sum_i \normgradwli}
	\end{eqnarray}
\end{proof}

\begin{algorithm}[t]
	\caption{Optimal Weighted SGD}
	\label{alg:optimalsgd}
	\KwIn{Initial $\vt{w_0}$, $T$}
	\KwOut{Final $\vt{w_T}$ }
	\For{$t = 1, ..., T$}{
		\ForEach{$i = 1, ..., n$}{
			$Grad[i] = \normgradwli$\;
		}
		$SumGrad = \sum_i Grad[i]$\;
		\ForEach{$i = 1, ..., n$}{
			$\pii$ = $Grad[i]/SumGrad$\;
		}
		sample $i$ from $\{1, ..., n\}$ based on distribution $\{p_1, ..., p_n\}$\;
		$\giw = \gradw \lwxyi / n\pii$\;
		$\vt{w_t} = \vt{w_{t-1}} - \eta \gradw \regw - \eta \giw$\;
	}
\end{algorithm}

Theorem~\ref{thm:main} can be viewed as a refined 
version of Theorem~\ref{thm:intuition} when label $y$ is observed. 
Theorem~\ref{thm:main} gives a rigorous
%%% ooibc: rigid --> rigorous == the meaning is very far apart
 explanation about our previous intuition from the variance reduction view of SGD optimization. 
 Note that this result makes no assumptions about the soft margin classification and applies to all sorts of ERM problems.
Algorithm~\ref{alg:optimalsgd} describes
 the optimal weighted SGD algorithm, 
 where the computation cost in each
  iteration will be optimized in the next subsection.

Another insight we observe from Theorem~\ref{thm:main} is
that in order to minimize the variance of SGD by using weighted sampling, 
$\giw$ should have the exact same magnitude for all instances.
%%% ooibc: pls check
\begin{eqnarray}
	\nonumber \uninorm{\giw} & = & \uninorm{\frac{\gradwli}{n\pii}}\\
	\nonumber & = & \frac{\uninorm{\gradwli} }{n * \frac{\normgradwli}{\sum_i \normgradwli}}\\
	& = & \frac{\sum_i \normgradwli}{n}
\end{eqnarray}
This suggests that only the \emb{direction} of the stochastic gradient is determined
by the training instance.
Further, the step size (i.e. the magnitude of $\eta \giw$) is of no consequence to
 the training instance and 
 is decided by the global learning rate.
 %%% ooibc: pls check
As a result,
the change of parameter 
in Active Sampler is much more steady than the
standard SGD.
This property agrees with the original purpose of gradient descent methods, 
as the gradient only indicates the fastest direction to minimize the objective loss function, 
without any indication on the step size.
%but does not suggest not how far the parameter should move. 
%%%% ooibc: the above does not make sense
%%   
Figure~\ref{fig:sgd} shows the comparison between standard SGD and Active Sampler.
For standard SGD, all training instances have the same sampling probability, while their gradient magnitudes vary. For Active Sampler, all training instances have the same gradient magnitude,
while their sampling probabilities vary.
Both methods have the same expectation of gradient,
however, Active Sampler has a smaller variance,
resulting in a faster and more stable convergence process.

\subsection{Practical Implementation Issues}
\label{sec:implementation}

\begin{figure}[t!]
	\centering
	\epsfig{width=8.5cm,file=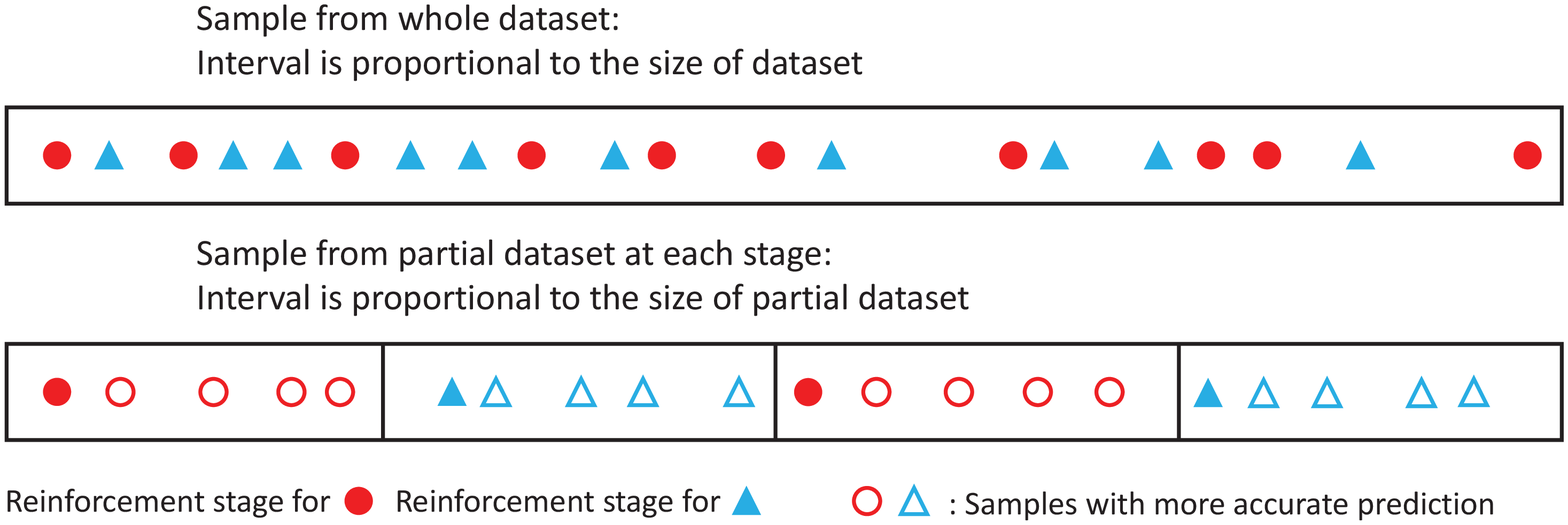}
	\caption{History Reinforcement}
	\label{fig:historyreinforcement}
	%\vspace{-0.2cm}
\end{figure}

As discussed in Section~\ref{sec:prel}, the main goal of optimizing the
SGD algorithm is to reduce the variance of $\giw$, while keeping
the computation cost per iteration light-weight.
We have already shown how to minimize the variance of $\giw$ by using the Active Sampler.
In this subsection,
we will discuss some practical issues in implementing Active Sampler onto real systems,
which may significantly affect the computation time in each iteration.

\subsubsection{Sampling based on History}

\label{sec:history}

The probabilistic distribution described in Theorem~\ref{thm:main} can indeed
minimize the variance of the stochastic gradient.  However,
use of the exact distribution, which requires all $n$ gradients to be evaluated
 at each step, is obviously not practical.
Instead, we can predict the gradient magnitude for each training instances using historical data. 
A straightforward approach to the problem is to remember the magnitude of the latest gradient of each instance 
and use it as an approximation. 
Considering that the actual gradient may change and the historical magnitude is only an approximation, 
a smoothing term is required. 
For example, if one instance contributes a zero gradient at any iteration of the model training when it is sampled,
 that sample will never be visited afterward if there is no smoothing, 
 notwithstanding this instance may become valuable for refinement of
  the model at later stages.

\begin{Definition}[History Approximation]
	\label{thm:prediction}
	Let $t_i$ be the latest step where training instance $i$ is 
   visited and let $\wti$ denote the parameter value at step $t_i$.
The sampling probability for each sample $i$ in a practical Active Sampler with smoothing is defined as:
	\begin{eqnarray}
		\pii = (1- \beta)\frac{\normgradwtili}{\sum_i \normgradwtili} + \frac{\beta}{n}
	\end{eqnarray}
\end{Definition}

\begin{algorithm}[t]
	\caption{ASSGD (Active Sampler SGD)}
	\label{alg:assgd}
	\KwIn{Initial $\vt{w_0}$, $T$, $\beta$, $Grad[]$, $SumGrad = \sum_i Grad[i]$}
	\KwOut{Final $\vt{w_T}$ }
	\For{$t = 1, ..., T$}{
		sample $i$ from $\{1, ..., n\}$ based on distribution $\{p_1, ..., p_n\}$ where $\pii = \beta/n + (1- \beta)Grad[i]/SumGrad$\;
		$\giw = \gradw \lwxyi / n\pii$\;
		$\vt{w_t} = \vt{w_{t-1}} - \eta \gradw \regw - \eta \giw$\;
		$SumGrad = SumGrad - Grad[i]$\;
		$Grad[i] = \normgradwli$\;
		$SumGrad = SumGrad + Grad[i]$\;
	}
\end{algorithm}

The scheme ensures that every training instance 
has at least $\beta$ times the average sampling probability (i.e. $1/n$) 
being sampled. 
%%% ooibc: pls check
Algorithm~\ref{alg:assgd} describes
the Active Sampler using the history approximation. 
In each iteration, only its gradient magnitude needs to be updated.

%%%Another fact that may be included in consideration is that:
We note that by using the history length to denote the number of iterations from the last time an instance
 is sampled, 
its history approximation becomes less accurate with the increase of the history length.
Meanwhile, 
the expectation of history length for one instance is $1/\pii$,
and the average of $\pii$ is $1/n$.
Consequently, the history approximation will become less accurate when data size becomes larger. 
To address this issue, we
 propose \ems{History Reinforcement}, whose key idea is illustrated in Figure~\ref{fig:historyreinforcement}. 
 History Reinforcement algorithm trains the model using a set of stages, 
 each of which contains a large amount of SGD iterations. 
 Within a stage, it first samples a subset of training data which consists of $m$ instances, and then uses them as the training set in its SGD iterations.
During the training of each stage, 
the sampling probability for the instances is $n/m$ times larger than training all the instances together. 
Therefore, the approximation will be much more accurate except 
the first time in a stage when one instance is sampled.
The only drawback of History Reinforcement is 
that it may lead to a bias in the training of a stage, 
as only partial data are trained.
However, \cite{minibatch} presents an effective scheme in an analogous context 
to reduce this bias by adding a regularizer to limit the change of parameter in one stage.
%%%% ooibc: the colon is use wrongly
Below, we formally define the concept of \emb{History Reinforcement}.

\begin{Definition}[History Reinforcement]
	\label{thm:history}
	History Rein\-for\-ce\-ment algorithm consists multiple stages. 
in each stage $t$, it draws a subset $I_t$ of training instances, which contains $m$ random instances from the whole dataset, and trains the model $\vt{w_t}$ using $g$ SGD iterations. The loss function used in each stage is:
	\begin{eqnarray}
		\losswt = \sum_{i \in I_t} \frac{\lwtixyi}{m} + \frac{\gamma_t}{2} \norm{\vt{w_{t-1}} - \vt{w_t}}
	\end{eqnarray}
	where $\gamma_t$ is a parameter in~\cite{minibatch} calculated based on $m$, $t$ and $\var{\giw}$.
\end{Definition}

The correctness and effectiveness of this batch training is given in Theorem~1 of \cite{minibatch} (by considering a stage as a batch step).
The average number of visits for one instance in a stage is $g/m$. 
Therefore, $1-m/g$ proportion of the iterations in a training stage will benefit from a more accurate approximation. 
In essence, there is a trade-off between the bias involved by using partial data and the accuracy gain in gradient approximation. 
Intuitively, for larger datasets, 
the bias becomes less significant while the accuracy gain by using History Reinforcement becomes more valuable.
On the contrary,
the approximation of gradient in small datasets is fairly accurate, and
therefore, directly sampling from the whole dataset is advantageous.

\begin{algorithm}[t]
	\caption{ASSGD with History Reinforcement}
	\label{alg:ashr}
	\KwIn{Initial $\vt{w_0}$, $T$, $m$, $g$}
	\KwOut{Final $\vt{w_T}$ }
	\For{$t = 1, ..., T$}{
		$I_t = \Phi$\;
		\For{$i = 1, ..., m$}{
			sample $t_i$ uniformly from $\{1, ..., n\} - I_t$\;
			$I_t = I_t \cup \{t_i\}$\;
		}
		Compute $\gamma_t$ based on \cite{minibatch}\;
		Train $\vt{w_t}$ using Algorithm~\ref{alg:assgd} for $g$ iterations, using $\vt{w_{t-1}}$ as initial $\vt{w_0}$, using $I_t$ as the training set, and using $\regw + \frac{\gamma_t}{2} \norm{\vt{w_{t-1}} - \vt{w}}$ as the regularization function\;
	}
\end{algorithm}

\subsubsection{Efficient Vectorized Computation}

To reduce the variance of the stochastic gradient, 
a widely adopted solution is to employ mini-batch training, 
which averages the stochastic gradients of multiple training samples. 
By averaging $b$ training samples, the variance of gradient can be reduced by $b$ times ($b$ is typically
 between 10 and 1000). 
 Meanwhile, thanks to the effect of vectorized computation 
 and the constant communication cost when the computations are parallelized, 
 %%%% ooibc: parallelized setting?  
 the training time per iteration for mini-batch SGD is much smaller 
 than $b$ times the training time per iteration for the standard SGD. 
 Therefore, mini-batch SGD is commonly used in most large-scale optimization problems. 
 Active Sampling is orthogonal to mini-batch SGD, so we could use both improvements simultaneously. 
 However, to integrate them together, we need to compute the average of $\giw$ for $b$ training samples in an efficient vectorized way, as well as to obtain the gradient magnitude for each training instance.

\begin{Definition}[Mini-batch SGD / Active Sampler]\label{def:minibatch} \quad \quad\\
	At each iteration $t$, mini-batch SGD uniformly draws $b$ samples $I_t = \{ t_1, ..., t_b\}$ from $\{ 1, ..., n\}$, and uses the averaged gradient as the stochastic gradient.
	\begin{eqnarray}
		\gtw & = & \sum_{i \in I_t} \frac{\gradw \lwxyi}{b}
	\end{eqnarray}
	At each iteration $t$, mini-batch Active Sampler repeats the sample selection in Theorem~\ref{thm:main} for $b$ times and get $b$ samples $I_t = \{ t_1, ..., t_b\}$, and uses the averaged gradient as the stochastic gradient.
	\begin{eqnarray}
		\gtw & = & \sum_{i \in I_t} \frac{\gradw \lwxyi}{b n \pii}
	\end{eqnarray}
	Similar to mini-batch SGD, the variance of $\gtw$ in mini-batch Active Sampler is reduced by $b$ times.
\end{Definition}

In mini-batch SGD, 
the main advantage stemmed from vectorized computation is that the actual gradients 
from all samples do not need to be stored individually and then aggregated. 
This trick is very time and memory efficient when the size of parameters is huge (e.g. deep neural network, sparse logistic regression). 
Here we use a multi-layer perceptron~\cite{systemdl2} (MLP) model to illustrate how the stochastic gradient of mini-batch SGD is computed, and how mini-batch Active Sampler can be computed in a similar light-weight manner.
Note that general linear models ($\fwx = \vt{w^T}\vt{x}$) are usually generalized as a multi-class classification problem, and their parameters $\vt{w}$ are also a matrix, 
which is similar to the hidden layer in MLP.
Therefore, general linear models can be viewed as a single layer perceptron with a small difference 
in the loss function and hence all the optimization techniques discussed below can be applied to these models as well.

\begin{Definition}[Multi-Layer Perceptron (MLP)]\label{def:mlp}
	Multi-Layer Perceptron~\cite{systemdl2} is a feed forward neural network.
	It consists of one input layer $H^{(0)}$, $h$ hidden layers ($H^{(k)}$, $k = 1 ..., h$ ) and a loss layer to compute the loss $\lwxyi$ based on the prediction $H^{(h)}$ for $\vt{x_i}$. 	
	Each hidden layer $k$ is a vector of units, and the calculation is formalized as follows:
	\begin{eqnarray}
		Z^{(k+1)} = W^{(k)}H^{(k)} + B^{(k)}\\
		H^{(k+1)}= \sigma(Z^{(k+1)})
	\end{eqnarray}
	where $\sigma(\cdot)$ is the activation function.
	The gradient is computed via back-propagation:
	\begin{eqnarray}
		\frac{\partial \lwxyi}{\partial W^{(k)}} = \frac{\partial \lwxyi}{\partial Z^{(k+1)}_i} \times {H^{(k)}_i}^T \\
		\frac{\partial \lwxyi}{\partial H^{(k)}} = \frac{\partial \lwxyi}{\partial Z^{(k+1)}_i} \times W^{(k)} \\
		\frac{\partial \lwxyi}{\partial Z^{(k)}_{p}} = \sigma'(Z^{(k)}_{p}) \frac{\partial \lwxyi}{\partial H^{(k)}_{p}}
	\end{eqnarray}
\end{Definition}

\begin{algorithm}[t]
	\caption{Batch Computation for Active Sampler}
	\label{alg:vectorizedminibatch}
	\KwIn{$H^{(k)}_{b \times l}$, $Z^{(k+1)}_{b \times m}$, $W^{(k)}_{m \times l}$, $\nabla H^{(k+1)}_{b \times m}$}
	\KwOut{$\nabla H^{(k)}_{b \times l}$, $\nabla W^{(k)}_{m \times l}$, $\uninorm{\grad{W^{(k)}} \lwxyi}$}
	\ForEach{$i \in \{0, ..., b-1\}$}{
		\ForEach{$p \in \{0, ..., m-1\}$}{
			$\nabla Z^{(k+1)}[i][p] = \sigma'(Z^{(k+1)}[i][p]) \nabla H^{(k+1)}[i][p]$\;
		}
	}
	$\nabla H^{(k)}_{b \times l} = \nabla Z^{(k+1)}_{b \times m} \times W^{(k)}_{m \times l}$\;
	$\nabla W^{(k)}_{m \times l} = \frac{1}{b} {(\nabla Z^{(k+1)})}^T_{m \times b} \times H^{(k)}_{b \times l}$\;
	// line 1-5 : compute stochastic gradient  $O(bml)$\\
	\ForEach{$i \in \{0, ..., b-1\}$}{
		$SumGZ = 0$, $SumH = 0$\;
		\ForEach{$p \in \{0, ..., m-1\}$}{
			$SumGZ = SumGZ + {(\nabla Z^{(k+1)}[i][p])}^2$\;
		}
		\ForEach{$q \in \{0, ..., l-1\}$}{
			$SumH = SumH + {(H^{(k)}[i][q])}^2$\;
		}
		$\uninorm{\grad{W^{(k)}} \lwxyi} = \sqrt{SumGZ * SumH}$
	}
	// line 7-13 : compute gradient magnitude $O(b(m+l))$\\
\end{algorithm}

We now analyze the computation of gradient for one layer $k$ in mini-batch SGD.
Using $m$ to denote the number of units in $H^{(k+1)}$, and $l$ to denote the number of units in $H^{(k)}$,
the parameter $W^{(k)}$ is an $m \times l$ matrix, and $\gt{W^{(k)}}$ is also an $m \times l$ matrix.
\begin{eqnarray}
	\nonumber {\gt{W^{(k)}}} & = & \sum_{i \in I_t} \frac{\grad{W^{(k)}} \lwxyi}{b}    \\
	& = & \frac{1}{b} \sum_{i \in I_t} \frac{\partial \lwxyi}{\partial Z^{(k+1)}_i} \times {H^{(k)}_i}^T %\\
	%& = & \frac{1}{b} \frac{\partial \lwxy}{\partial H_{j+1}} \times H_j
\end{eqnarray}
However, directly computing the $b$ gradients $\frac{\partial \lwxyi}{\partial Z^{(k+1)}_i} \times {H^{(k)}_i}^T$ one by one is not efficient, as each gradient is an $m \times l$ matrix.
Instead, using $H^{(k)}$ to denote the $b \times l$ lower layer feature matrix $[H^{(k)}_1, ..., H^{(k)}_b]^T$, and using $\frac{\partial \lwxy}{\partial Z^{(k+1)}}$ to denote the $b \times m$ higher layer gradient matrix $[\frac{\partial \lwxyi}{\partial Z^{(k+1)}_1}$, $...$, $\frac{\partial \lwxyi}{\partial Z^{(k+1)}_b}]^T$, we have:
\begin{eqnarray}
	\nonumber {\gt{W^{(k)}_{pq}}} & = & \frac{1}{b} \sum_{i \in I_t} \frac{\partial \lwxyi}{\partial Z^{(k+1)}_{ip}} \times H^{(k)}_{iq} \\
	\Rightarrow \quad \quad \quad {\gt{W^{(k)}}} & = & \frac{1}{b} {{(\frac{\partial \lwxy}{\partial Z^{(k+1)}})}^T} \times {H^{(k)}}
\end{eqnarray}
Therefore, it is computed by performing
matrix multiplication for an $m \times b$ matrix and a $b \times l$ matrix, which is obviously
more efficient than the previous method,
which computes multiple vector-vector multiplications.
% since matrix-matrix multiplication is much faster than vector-vector multiplication.  Jag changed to line above. 
It also reduces the memory cost from $b \times m \times l$ to $m \times l$. 
The computation for $\frac{\partial \lwxyi}{\partial H^{(k)}}$ is analogous. 
%since matrix-matrix multiplication is also faster than matrix-vector multiplication.

In mini-batch Active Sampler, 
there are two differences compared to mini-batch SGD. 
First, Active Sampler needs to provide
 each instance a weight based on $1/n \pii$. 
 Second, Active Sampler needs to compute the gradient magnitude for every training instance. 
 For the first problem, the solution is quite straightforward -- putting the weight in the loss function before calculating its gradient for the parameters. 
 By scaling the value of loss by $1/n \pii$ times, its gradient will change $1/n \pii$ times accordingly.
 
\begin{table}[t!]\caption{Datasets and Models}
	\label{tbl:dataset}
	\vspace{-0.3cm}
	\center{
		\scriptsize
		\begin{tabular}{|c|c|c|c|c|}
			\hline
			Dataset & \# Examples & Size & Model & Test Error\\
			\hline\hline
			MNIST & 60K & 57MB & kernel SVM & 0.6\% \\
			\hline
			URL & 2.4M & 950MB & Lasso & 2.5\% \\
			\hline
			CIFAR10 & 60K & 161MB & DCNN &  18\% \\
			\hline
			CIFAR-DA & 7.6M & 14.8GB & DCNN & 11.5\% \\
			\hline
		\end{tabular}
	}
	\vspace{-0.25cm}
\end{table}

For the calculation of the gradient magnitude, 
we exploit the following equation to avoid explicitly calculating the gradient of each training instance:
\begin{eqnarray}
	\nonumber & \norm{\grad{W^{(k)}} \lwxyi} \\
	\nonumber = & \sum_{p \in m} \sum_{q \in l} \sqr{\frac{\partial \lwxyi}{\partial W^{(k)}_{pq}}} \\
	\nonumber = & \sum_{p \in m} \sum_{q \in l} \sqr{\frac{\partial \lwxyi}{\partial Z^{(k+1)}_{ip}} H^{(k)}_{iq}}\\
%	\nonumber = & \sum_{p \in m} \sum_{q \in l} {\frac{\partial \lwxyi}{\partial Z^{(k+1)}_{ip}}}^2 {H^{(k)}_{iq}}^2 \\
	= & (\sum_{p \in m} {\frac{\partial \lwxyi}{\partial Z^{(k+1)}_{ip}}}^2)(\sum_{q \in l} {H^{(k)}_{iq}}^2)
\end{eqnarray}
Therefore, we just need to compute the product of the square sum of $\frac{\partial \lwxyi}{\partial Z^{(k+1)}_i}$ and $H^{(k)}_i$, which are all from intermediate results during the computation of mini-batch SGD. 
Its computation complexity is just $O(b(m+l))$, which is extremely light-weight considering that the cost for calculating the gradient is $O(bml)$.
Algorithm~\ref{alg:vectorizedminibatch} shows the vectorized computation of Active Sampler in each layer of MLP. For deep models which contains multiple layers, the square of the gradient magnitude with respect to parameters from whole layers can be computed by summing the square of gradient magnitude with respect to parameters from each layer, i.e.
\begin{eqnarray}
	\norm{\gradw \lwxyi} & = & \sum_k \norm{\grad{W^{(k)}} \lwxyi} \quad \quad
\end{eqnarray}

\section{Experimental Study}
\label{sec:exp}

\subsection{Experiment Setup}

\begin{figure*}
	\centering
	\subfloat[MNIST]{
		\epsfig{width=.24\textwidth,file=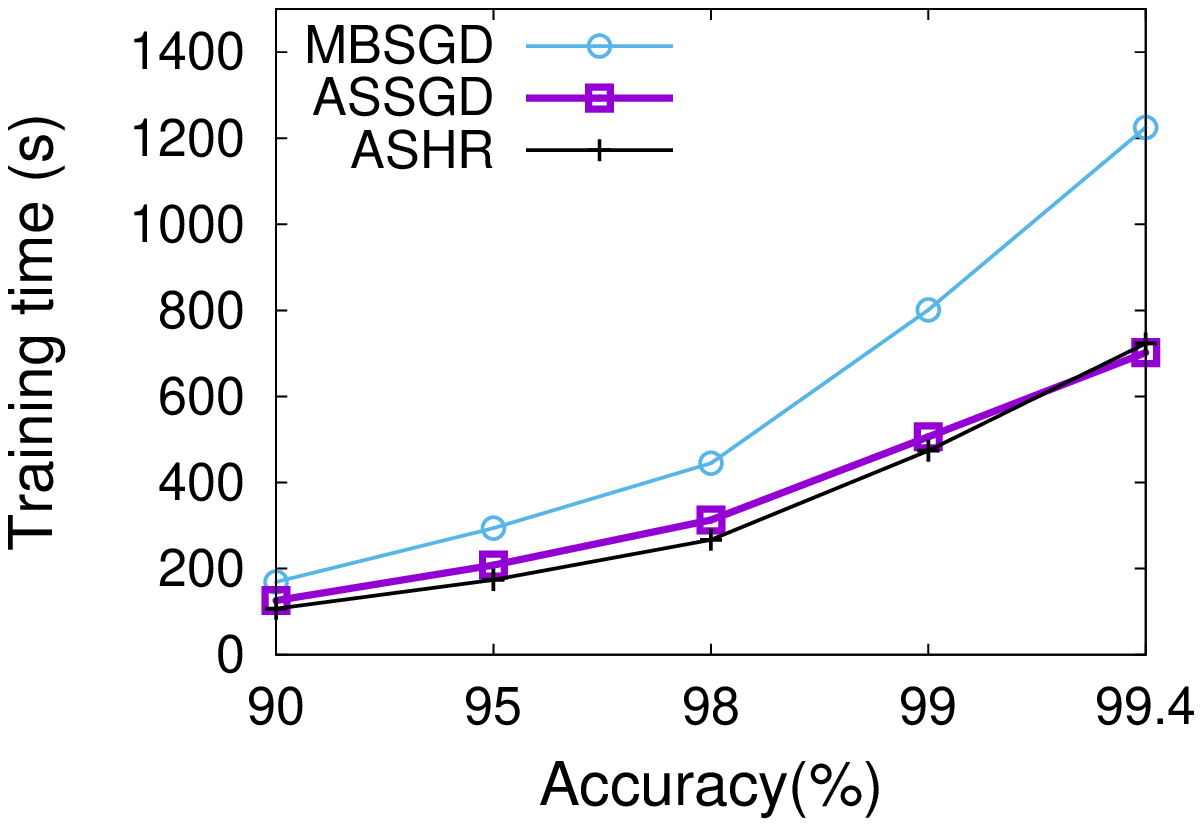}}
	\hspace{0mm}\subfloat[URL]{
		\epsfig{width=.24\textwidth,file=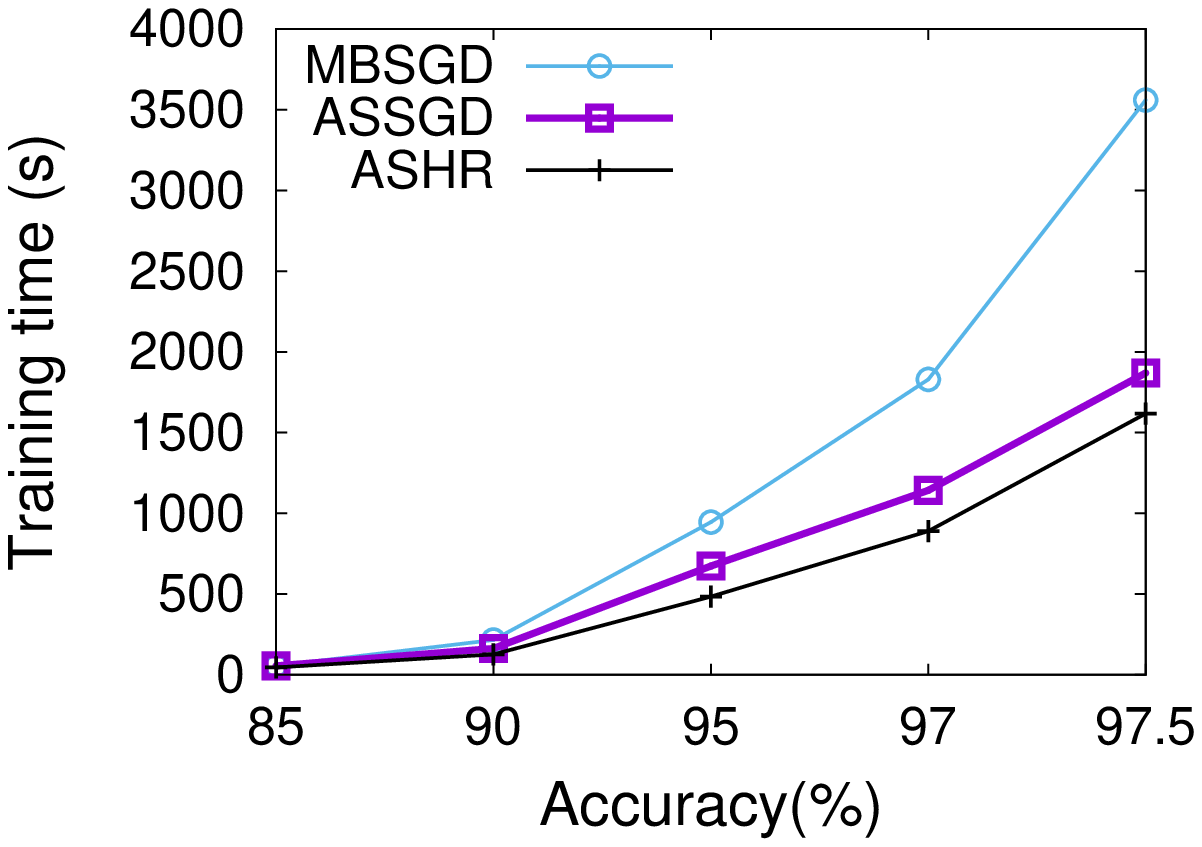}}
	\hspace{0mm}\subfloat[CIFAR10]{
		\epsfig{width=.24\textwidth,file=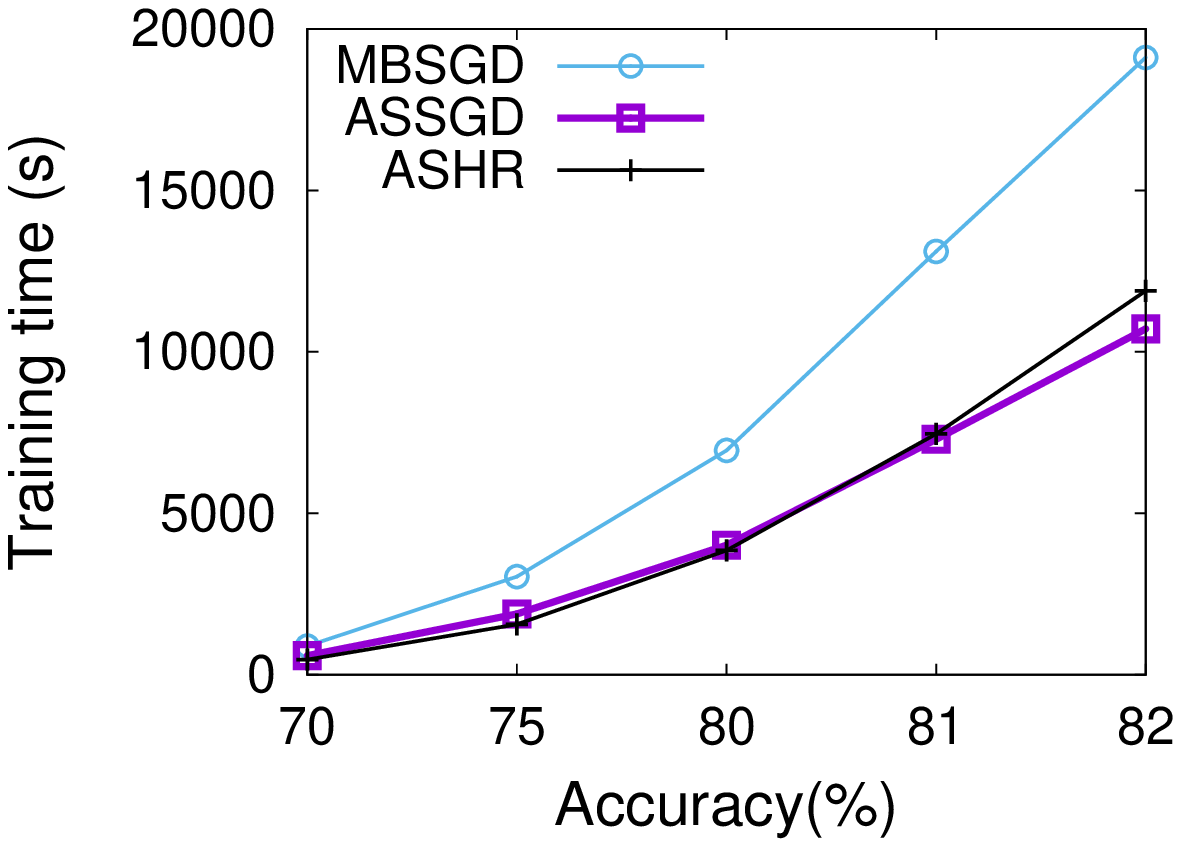}}
	\hspace{0mm}\subfloat[CIFAR-DA (Distributed)]{
		\epsfig{width=.24\textwidth,file=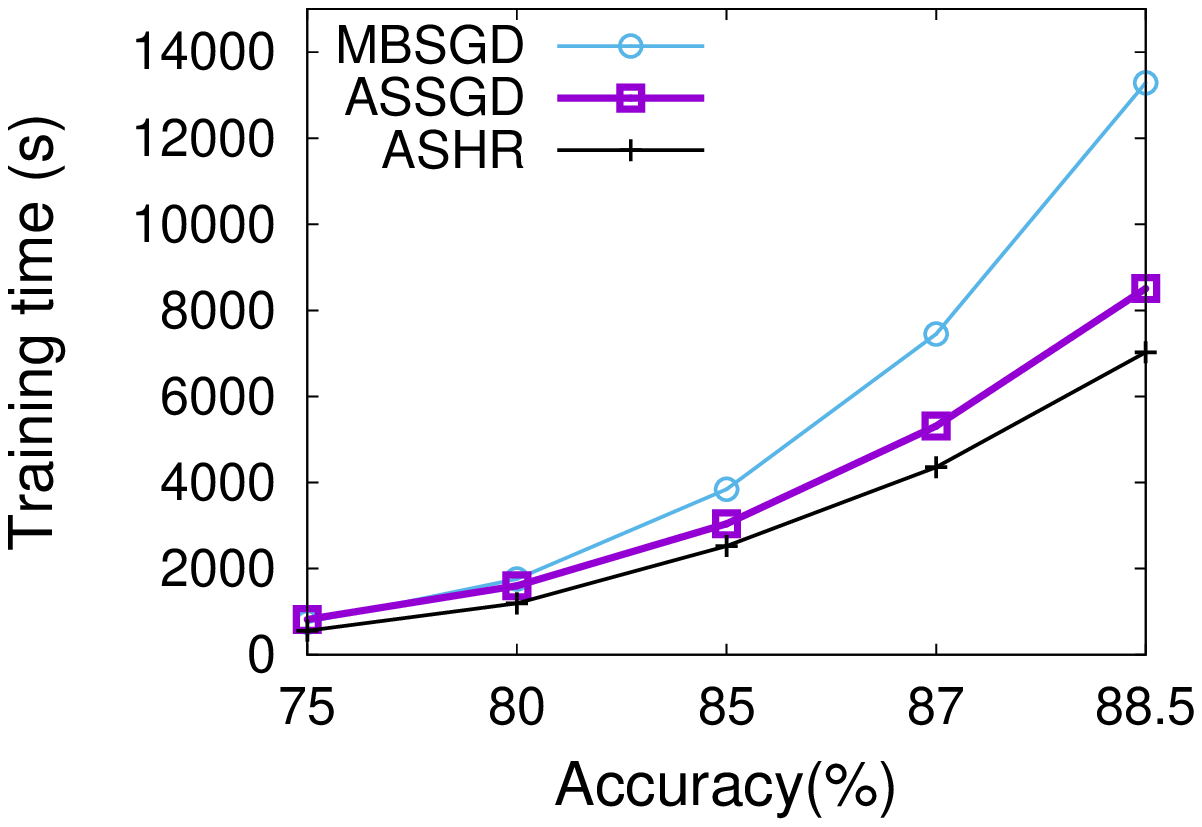}}
	\caption{Overall Training Time}
	\label{fig:overall}
	\vspace{-0.20cm}
\end{figure*}

We shall evaluate the speedup of Active Sampler using
 a set of popular benchmark tasks,
 namely the hand-written digit classification on MNIST using SVM~\cite{pegasos},
 malicious URL detection in URL using feature selection~\cite{largelasso},
  and image classification on CIFAR10 using CNN~\cite{cnn}.  
In addition, 
we shall also test the scalability of Active Sampler using the CIFAR10 dataset with data augmentation, 
  where the size of the training data is increased by 128x.  
Table~\ref{tbl:dataset} summarizes the datasets and models used in our experimental study.

\squishlist

\item \emb{MNIST:} MNIST is a benchmark dataset of handwritten digits classification, 
consisting of a training set of 60000 images and a test set of 10000 images. 
Each image contains 28*28 gray pixels. Pegasos~\cite{pegasos} is a mini-batch SGD solver for kernel SVM model. The test error of kernel SVM in MNIST dataset is 0.60\%~\cite{pegasos}.

\item \emb{URL:} URL~\cite{url} is a dataset for malicious URL detection. 
It consists of 2.4 million URLs and 3.2 million features. Each URL contains around 100 non-zero features and hence its features are quite sparse. Lasso regression~\cite{largelasso} is a popular feature selection model as described in Table~\ref{tbl:erm}. 
Its test error in the URL dataset is around 2.5\%.

\item \emb{CIFAR10:} CIFAR10 is a dataset for image classification, consisting of a training set of
60000 images and a test set of 10000 images. 
It is the benchmark dataset commonly used
for the evaluation of deep convolutional neural network (DCNN)~\cite{cnn} models. 
Each image contains 32*32 colored pixels. 
Its test error without data augmentation is 18\%.

\item \emb{CIFAR-DA:} Data augmentation is a standard technique to increase the size of training data.  
It generates additional images by slightly translating the original images.
We use the data augmentation version of the CIFAR10 dataset (CIFAR-DA) to study the scalability of Active Sampler. It contains 128x images compared with CIFAR10. Its test error for DCNN model is 11.5\%. However, as the number of training images increases, DCNN model takes significantly longer time to achieve its best performance.

\squishend

All the models are trained under the SGD framework. The standard mini-batch SGD (\emb{MBSGD}) algorithm is used as the baseline. 
We also implement the mini-batch version of Active Sampler for comparison, 
with the same size of mini-batch. 
Unless otherwise specified, the size of mini-batch is set to 128. 
The validation accuracy are tested per 100 mini-batch iterations. To study the effect of History Reinforcement strategy described in Section~\ref{sec:history}, we have implemented two versions of Active Sampler,
with and without History Reinforcement.
\emb{ASSGD}, the Active Sampler without History Reinforcement,
is expected to to perform well for moderate size of training examples,
 while \emb{ASHR}, the Active Sampler with History Reinforcement,
  is expected to yield better performance for large-scale training sets.
In ASHR, the whole dataset is randomly split 
into 16 large batches, and examples are trained 16 times on average at each stage.

All of the algorithms are implemented in C++, compiled using GCC O2, 
and OpenBlas is adopted to accelerate linear algebra operations. 
Experiments for MNIST, URL and CIFAR10 are carried on an Intel Xeon 24-core server with 500GB memory.

\emb{Distributed Environment and Scalability Test:}\\
We also study the performance of Active Sampler in a
distributed environment with the scale of millions of training examples. 
In general, 
the benefit of Active Sampler is independent
to the architecture of the training system 
as long as it is still under the SGD framework.
SGD training algorithms can be easily distributed to clusters via the parameter server~\cite{paraserver} architecture. %%%
The main difference between a distributed SGD and a single-node SGD is 
that the distributed SGD incurs additional communication overhead.
Since the communication cost for a mini-batch is a constant 
while the computation cost is proportional to its size, 
distributed SGD usually uses a larger mini-batch to reduce the proportion of communication cost.
We conduct our scalability study using CIFAR-DA dataset on the Apache SINGA system~\cite{singa},
which is a general distributed deep learning platform.
%%% training for DCNN and other deep learning models. 
We follow all the default settings of CIFAR on SINGA, where the mini-batch size is set to be 512. 
The distributed environment is a 32-node cluster, where each machine is equipped with an Intel Xeon 4-core CPU and 8GB memory.

\subsection{Overall Performance}

%%%% commented
\eat{Figure~\ref{fig:least-most} shows the
training images in MNIST which are least and most sampled. 
Those least sampled images are obviously
 easier to classify than those most sampled images, for both machines or humans. 
The results demonstrate that our active learning based intuition is 
well carried out in the practical implementation -- 
examples that are ``hard to classify'' are indeed sampled more frequently. 
Next, we directly measure how much Active Sampler can speedup the training process.}

Figure~\ref{fig:overall} shows the training time to 
reach a certain accuracy for MBSGD, ASSGD and ASHR.
99.4\%, 82\%, 97.5\% and 88.5\% are the best accuracy achieved in these four tasks respectively.
Generally, the convergence of ASSGD and ASHR are significantly faster than MBSGD. 
To reach the optimal test error, 
ASSGD and ASHR save about 40\% to 60\% of the training time. 
The speedup is especially great in the latter stages of training, as evidenced by the bigger difference in the slope between the algorithms in the right hand portion of each graph.
A possible explanation to this phenomenon is that the models typically
have smaller changes near the end stage of training. 
As a result, 
the larger variance in MBSGD would have more serious negative effect, 
while ASSGD and ASHR algorithm would get even better approximation of the scale of gradient. 
There are also some notable differences
between the performance of ASSGD and ASHR. 
First, ASHR converges much faster than ASSGD in the two large datasets (URL and CIFAR-DA), 
demonstrating that ASHR provides more accurate approximation of the
scale of gradient in large-scale datasets. 
Meanwhile, its speed-up is less than ASSGD in the two small datasets (MNIST and CIFAR10), 
probably due to only a subset of training data
used by ASHR which only contains around 3000 training examples. 
Second, ASHR converges slightly faster at the beginning,
and slightly slower near the end. 
This is because ASSGD needs to visit the whole dataset at least once before 
enjoying the benefit of smaller variance, 
while ASHR only needs to visit a subset of the dataset. 
However, in later stages of training, ASSGD gets the gradient approximation as accurate as ASHR since the model change is not significant, while ASHR is still suffering from the bias introduced by sampling from partial data.

For the scalability test on CIFAR-DA, 
its training time is even smaller than CIFAR10 due to distributed training. 
Active Sampler still works as expected: ASHR speeds up the training process by 1.9x, and ASSGD speeds up the process by 1.6x. 
This is because the benefit of Active Sampler is derived
from the total number of iterations used to achieve a certain accuracy,
instead of the improvement of training time per iteration.
Since the total number of iterations used is not affected by the training architecture, 
the speed-up 
of Active Sampler is applicable to all kinds of SGD frameworks as long as its overhead 
in the training time per iteration is small. ,
Conversely,
the number of training examples does affect the performance of ASSGD, 
since its approximation becomes less accurate when the number of training examples increases.
However, ASHR scales well in all cases.
%%% ooibc: pls check

\subsection{Variance of Stochastic Gradients}

\begin{figure*}
	\centering
	\subfloat[MNIST]{
		\epsfig{width=.24\textwidth,file=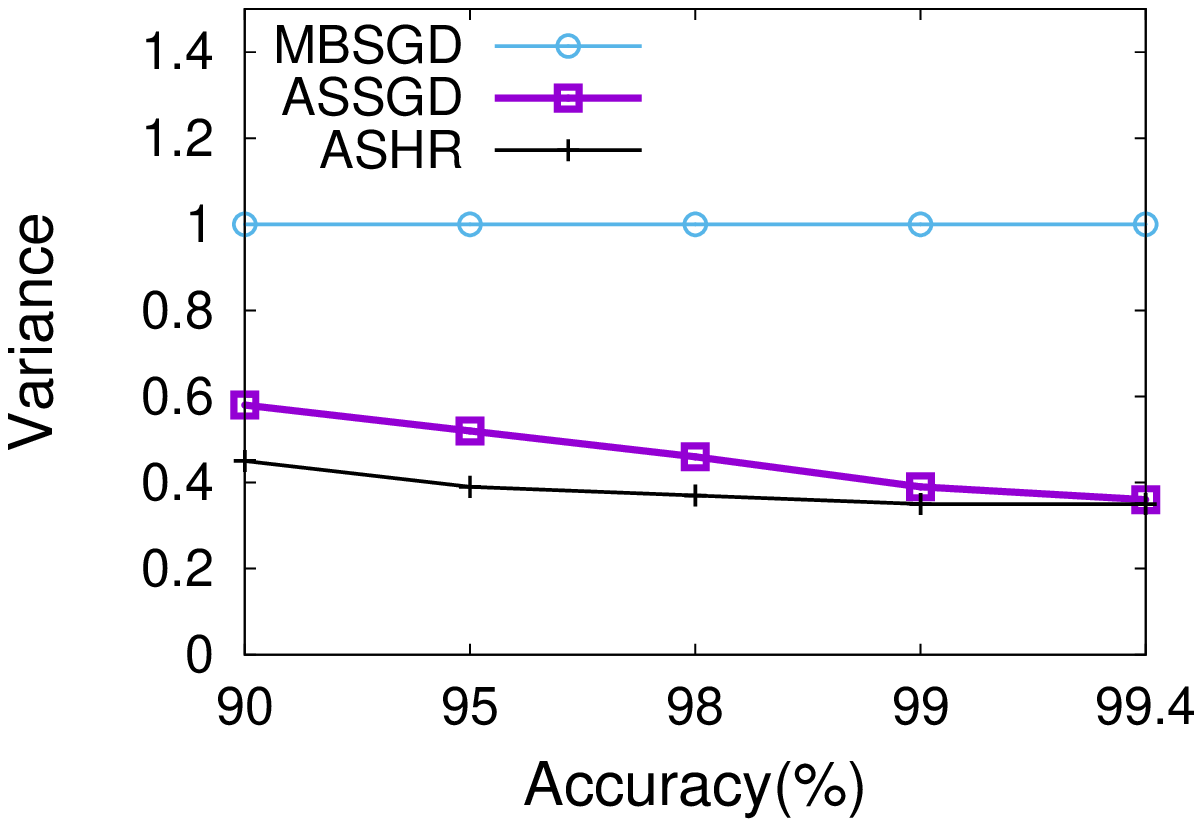}}
	\hspace{0mm}\subfloat[URL]{
		\epsfig{width=.24\textwidth,file=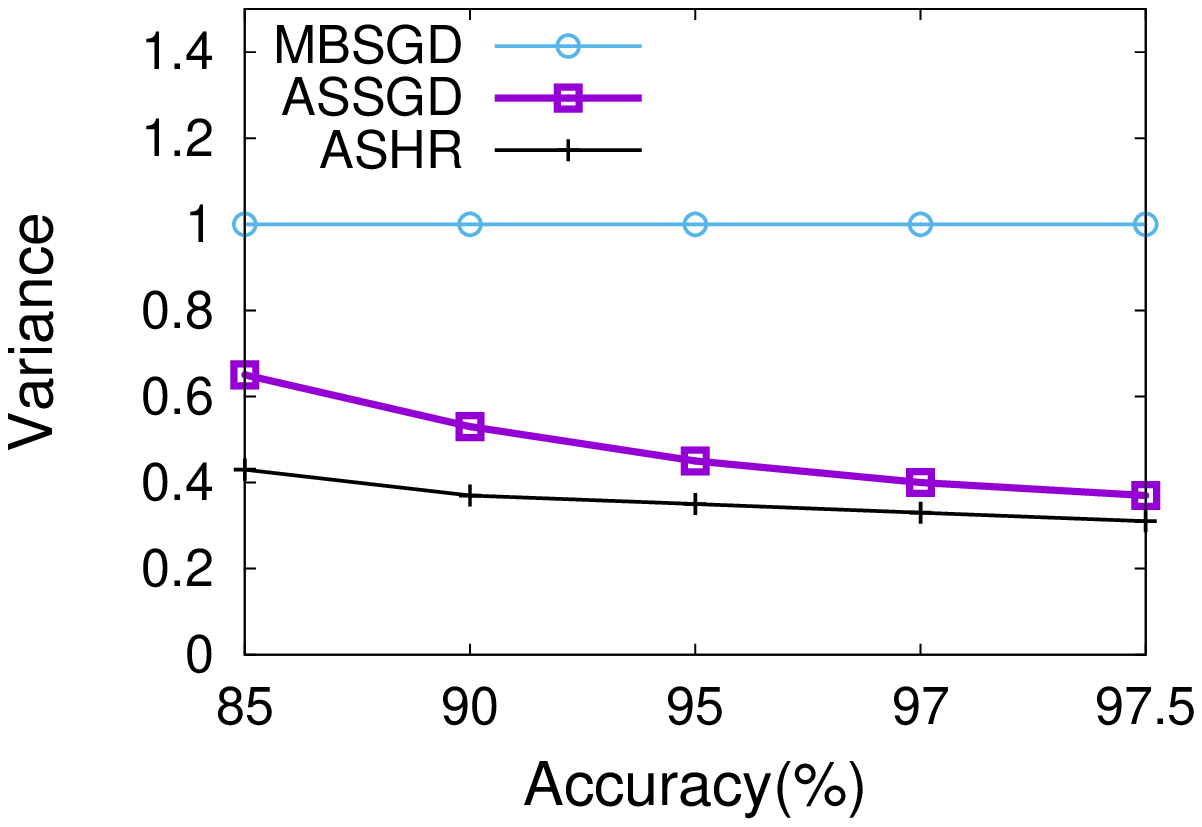}}
	\hspace{0mm}\subfloat[CIFAR10]{
		\epsfig{width=.24\textwidth,file=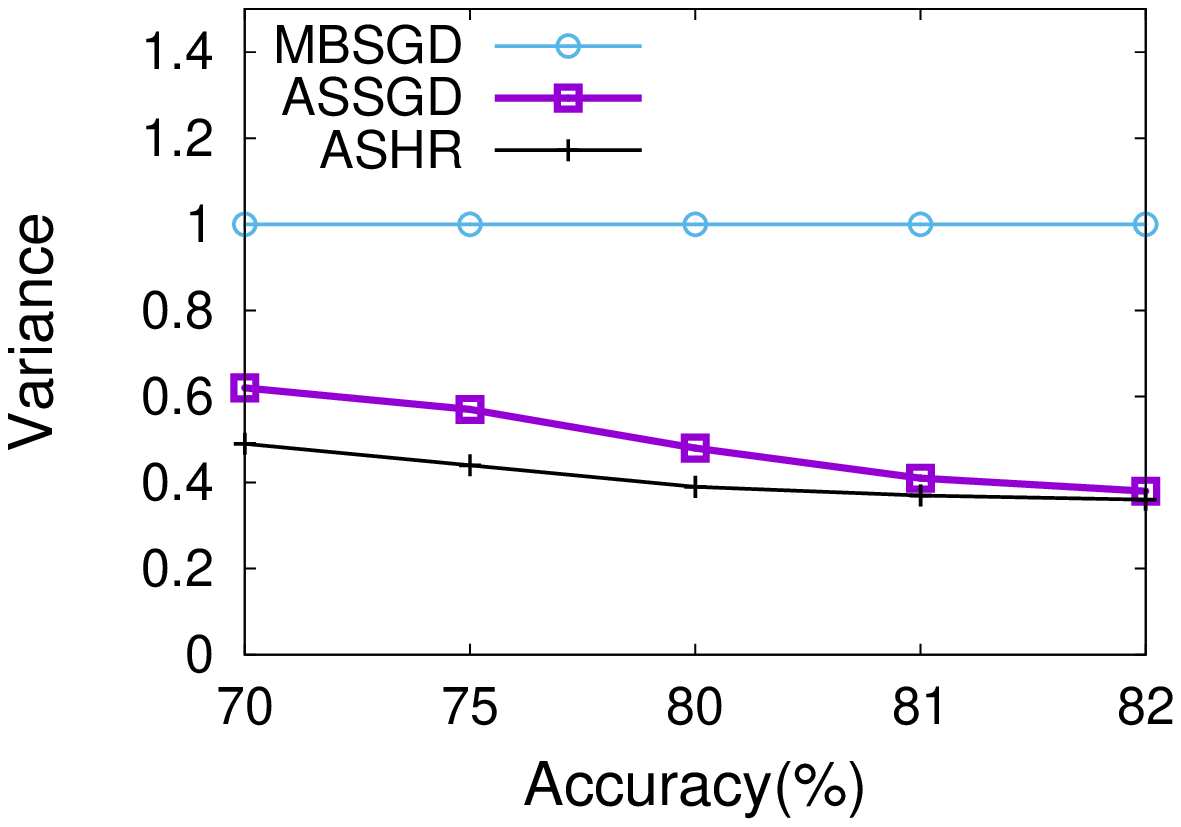}}
	\hspace{0mm}\subfloat[CIFAR-DA (Distributed)]{
		\epsfig{width=.24\textwidth,file=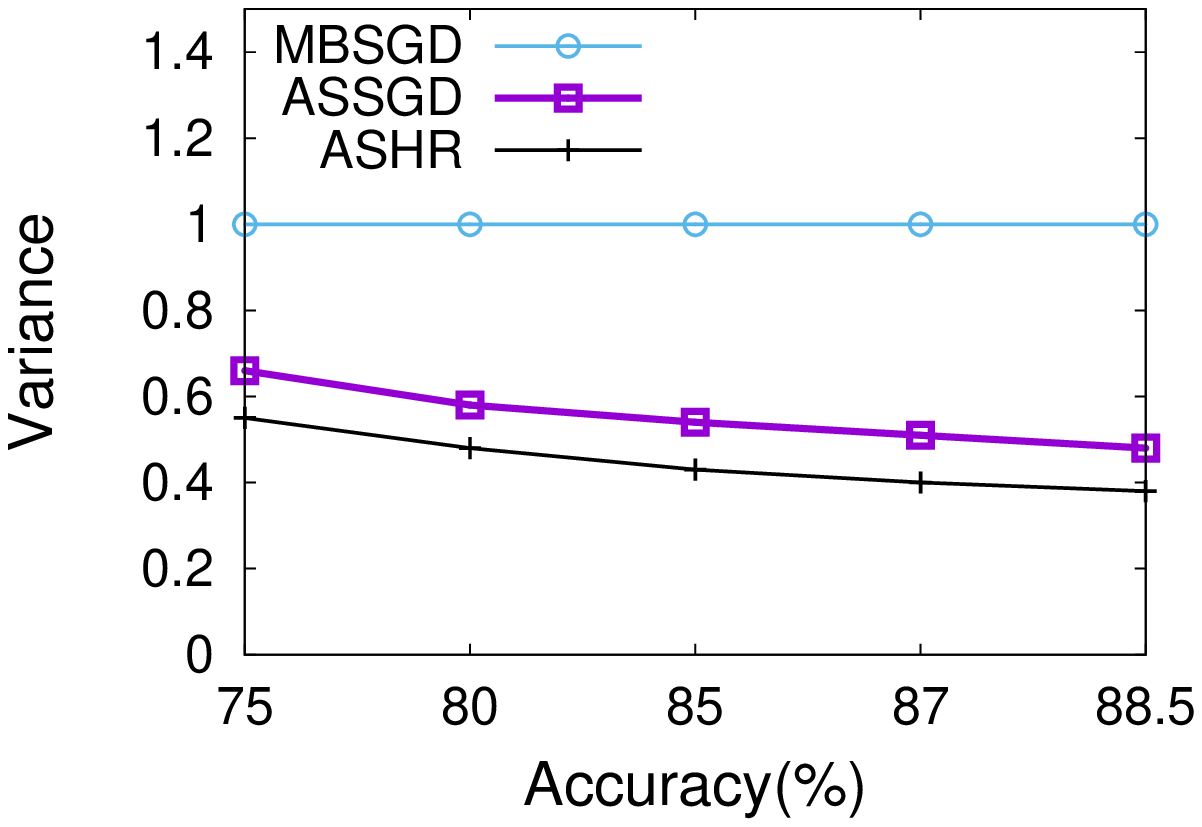}}
	\caption{Variance at Different Training Stages}
	\label{fig:variance}
	\vspace{-0.30cm}
\end{figure*}

From the stochastic optimization view point, 
the benefit of using Active Sampler 
is derived
mainly from the reduction of the variance of stochastic gradient. 
We therefore evaluate the average variance of MBSGD, ASSGD and ASHR 
and summarize the results in Figure~\ref{fig:variance}.
Since the absolute value of variance may change dramatically during the training process, 
we use the variance of MBSGD as the baseline 
and report the relative ratio of the variance of ASSGD and ASHR compared with the baseline. 
The results show that ASHR has the smallest variance, less than 40\% of the variance of MBSGD on average. 
The variance of ASSGD is slightly higher than that of ASHR, 
especially in the two large datasets, URL and CIFAR10, 
mainly because that its gradient approximation is less accurate in larger datasets. 
However, the variance of ASSGD is still less than half of the variance of MBSGD 
with the same min-batch size.
Another observable trend 
is that the variance ratio of ASSGD and ASHR are getting smaller with the increase of training time, 
suggesting that the history gradient approximation is getting more and 
more accurate as the training time increases.

We note that in MBSGD, 
the variance is proportional to the reverse of the size of mini-batch. 
Therefore, to get the same level of variance in the stochastic gradient as Active Sampler, 
MBSGD needs to increase its min-batch size by 2-3x. 
From this angle, 
Active Sampler is a much more efficient method
 to reduce the variance of stochastic gradient, instead of relying on the use of
 a larger mini-batch.
%%%% ooibc: the para is not well written
%%%%  need to come back

\subsection{Training Time Analysis}

%\begin{figure}
%	\centering
%	\epsfig{width=6cm,file=plt/time_per_iter.eps}
%	\caption{Training Time per Iteration}
%	\label{fig:time-per-iteration}
%\end{figure}

\begin{table}[t!]\caption{Training Time per Iteration}
	\label{tbl:time-per-iteration}
	\vspace{-0.3cm}
	\center{
		\begin{tabular}{|c|c|c|c|}
			\hline
			Dataset & MBSGD & ASSGD & ASHR\\
			\hline\hline
			MNIST & 0.179s & 0.208s & 0.205s \\
			\hline
			URL & 0.080s & 0.092s & 0.092s  \\
			\hline
			CIFAR10 & 0.245s & 0.295s & 0.296s  \\
			\hline
			CIFAR-DA & 0.110s & 0.119s & 0.119s  \\
			\hline
		\end{tabular}
	}
	\vspace{-0.2cm}
\end{table}

%\begin{figure*}[!t]
%	\centering
%	\begin{minipage}{.22\textwidth}
%		\centerline{\epsfig{file = plt/mnist_iter.eps, scale=0.4}}
%		\caption{\scriptsize MNIST}
%	\end{minipage}
%	\quad
%	\begin{minipage}{.22\textwidth}
%		\centerline{\epsfig{file = plt/url_iter.eps, scale=0.4}}
%		\caption{\scriptsize URL}
%	\end{minipage}
%	\quad
%	\begin{minipage}{.22\textwidth}
%		\centerline{\epsfig{file = plt/cifar10_iter.eps, scale=0.4}}
%		\caption{\scriptsize CIFAR10}
%	\end{minipage}
%	\quad
%	\begin{minipage}{.22\textwidth}
%		\centerline{\epsfig{file = plt/cifar_da_iter.eps, scale=0.4}}
%		\caption{\scriptsize CIFAR-DA}
%	\end{minipage}
%	%\vspace{-5pt}
%	\caption{Number of Iterations to Reach a Certain Accuracy}
%	\label{fig:iteration-accuracy}
%	\vspace{-0.15cm}
%\end{figure*}

%\begin{figure*}
%	\centering
%	\subfloat[MNIST]{
%		\epsfig{width=.24\textwidth,file=plt/mnist_iter.eps}}
%	\hspace{0mm}\subfloat[URL]{
%		\epsfig{width=.24\textwidth,file=plt/url_iter.eps}}	
%	\caption{Number of Iterations to Reach a Certain Accuracy}
%	\label{fig:iteration-accuracy}
%\end{figure*}

The overall training time is determined
by the product of the training time per iteration and the number of iterations to reach a certain accuracy. 
Here we present
a detailed study of how Active Sampler affects these two aspects.

Table~\ref{tbl:time-per-iteration} shows the training time 
per iteration of MBSGD, ASSGD and ASHR. 
Obviously, MBSGD is the fastest since ASSGD and ASHR entail additional computations. 
However, the difference is not significant. 
ASSGD and ASHR only require 15\%-20\% more time than MBSGD, 
while providing the stochastic gradient with much smaller variance.
As discussed above, 
to reach the same variance, 
MBSGD needs to use 2-3x samples in a mini-batch, which may incur 100\%-200\% additional overhead. 
In the distributed trained task CIFAR-DA, 
the overhead of ASSGD and ASHR are even smaller, 
which are around 10\%. 
This is because ASSGD and ASHR do not incur any overhead to the communication costs. 
There are no major differences between ASSGD and ASHR, 
since they have exactly the same computation logics
 inside each iteration.

\begin{figure*}
	\centering
	\subfloat[MNIST]{
		\epsfig{width=.24\textwidth,file=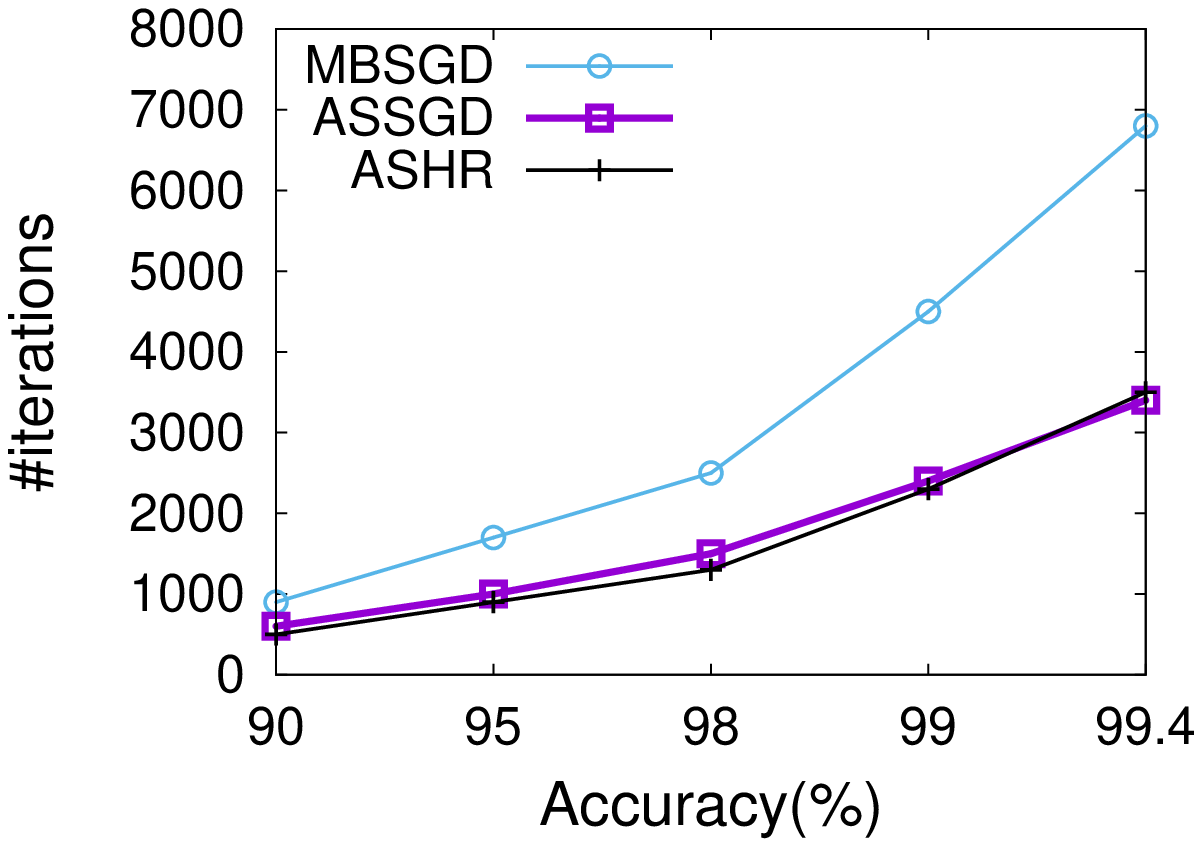}}
	\hspace{0mm}\subfloat[URL]{
		\epsfig{width=.24\textwidth,file=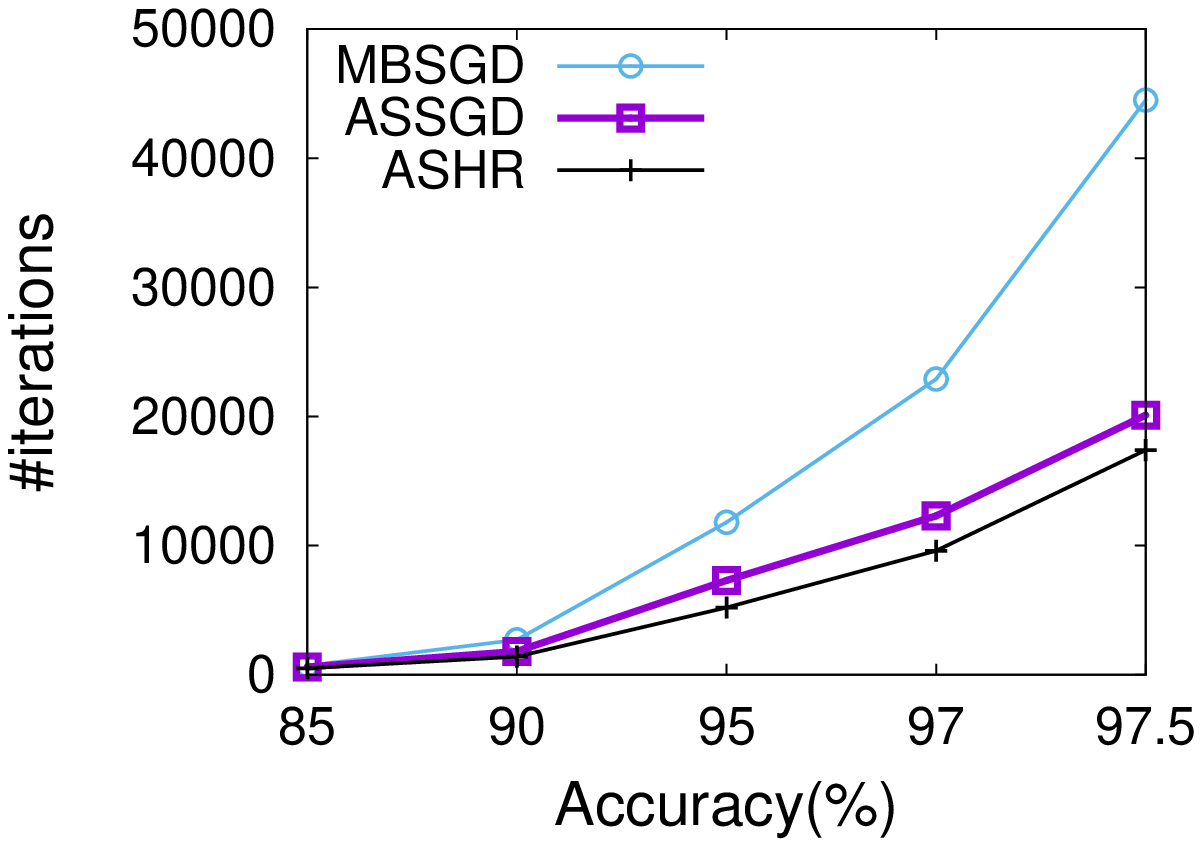}}
	\hspace{0mm}\subfloat[CIFAR10]{
		\epsfig{width=.24\textwidth,file=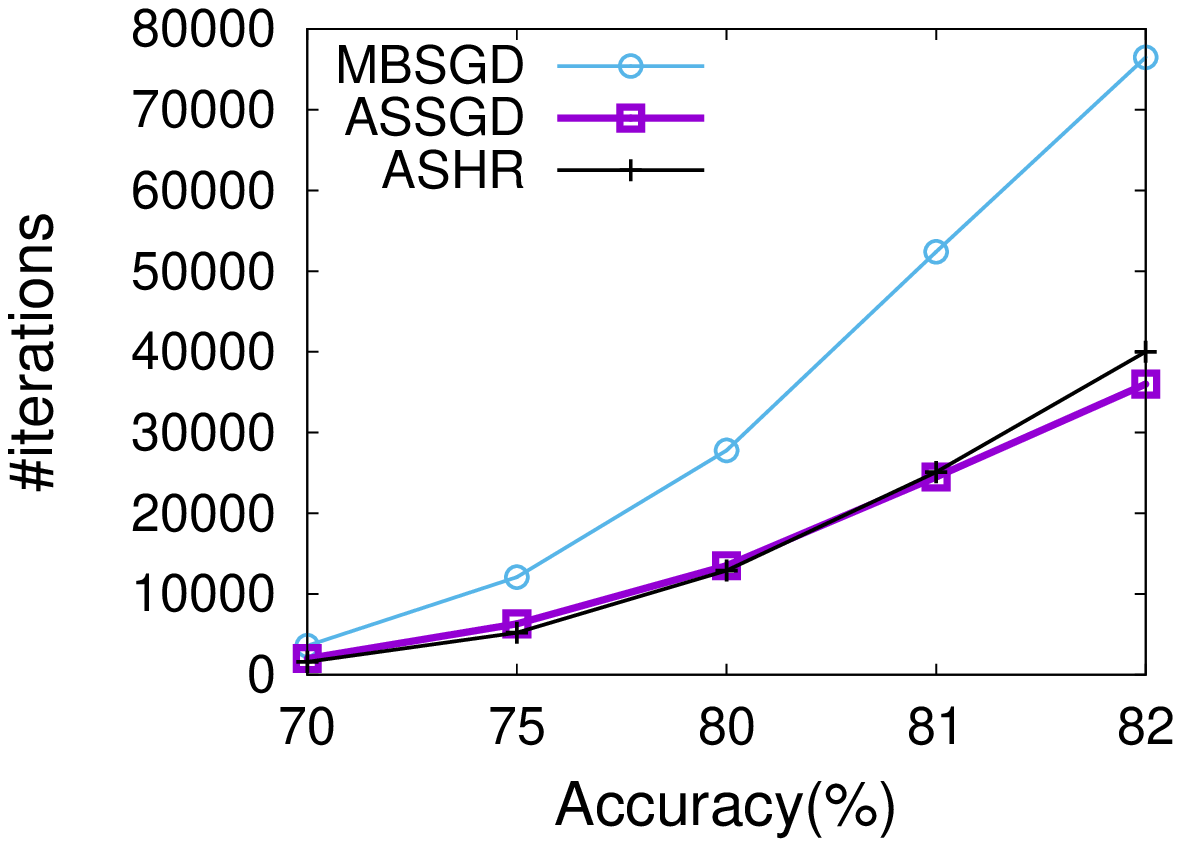}}
	\hspace{0mm}\subfloat[CIFAR-DA (Distributed)]{
		\epsfig{width=.24\textwidth,file=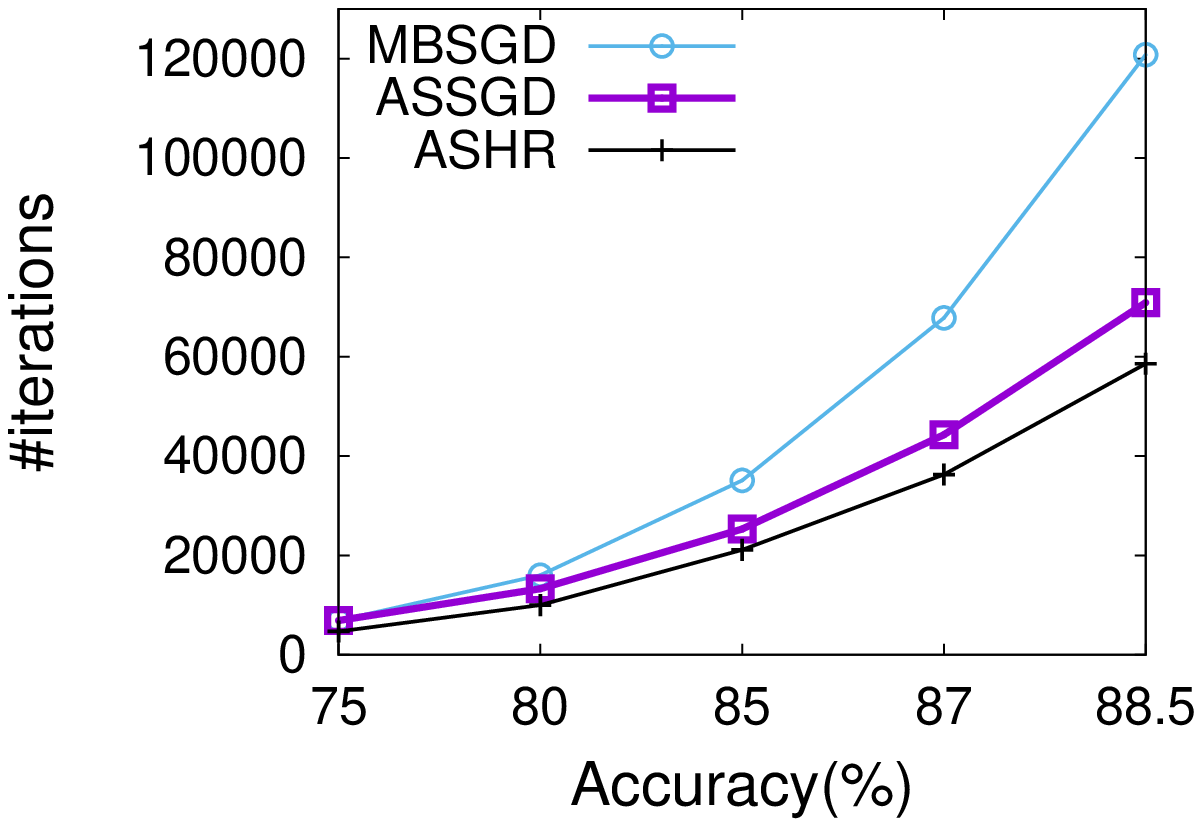}}
	\caption{Number of Iterations with respect to Accuracy}
	\label{fig:iteration-accuracy}
	\vspace{-0.15cm}
\end{figure*}

Figure~\ref{fig:iteration-accuracy} shows the number of iterations to reach a certain accuracy. 
The number of iterations required by ASSGD and ASHR is around 40\% to 60\% of the number of iterations required by MBSGD. The proportion of iterations saved varies with different datasets. 
The iterations saved would be more significant when the contribution from training examples are highly biased. However, theoretically, it is possible that all examples have the similar effect on refining the model (an extreme case is all examples being the same), and uniform sampling becomes the optimal weighted sampling. Not surprisingly, as indicated by the experiments so far, all of the benchmark datasets do not represent the extreme case, and a significant number of training iterations can be saved.

\section{Related Work}
\label{sec:related}

Complex machine learning models, such as large-scale linear methods\cite{pegasos}, 
feature selection~\cite{largelasso} or deep learning~\cite{systemdl2}, 
are widely adopted in Big Data analytics. 
Due to the huge size of both model and data, 
how to train these model efficiently is a challenging topic, 
and the solution requires efforts from learning, database, and system communities.
Many optimizations have been proposed from a systems perspective
for specific classes of models~\cite{systemfs,systemgs,systemglm, systemdl1, systemdl2, systemmd}. 
Most of these algorithms (and many others) can fit into an Empirical Risk Minimization~\cite{statisticalml} (ERM) framework, 
for which we aim to develop a more general accelerator.

The optimization of the general ERM is widely studied in machine learning community~\cite{statisticalml}. 
Generally, there are two classes of methods: 
first-order algorithms such as gradient descent~\cite{gdsurvey}, and second-order algorithms such as Newton method~\cite{newton1}. 
Although second-order algorithms typically
have a much faster convergence rate, 
they require the Hessian matrix~\cite{boydcvx} of parameters, 
%%%% ooibc: reference for Hessian matrix?
making them not practical
 for large-scale models where the number of parameter is huge. 
For similar reasons, 
batch gradient methods~\cite{batchgd} are very expensive for
large training datasets. 
Therefore, stochastic methods~\cite{sgd} are the most favored algorithm in
recent large-scale machine learning applications.
\eat{There are also some works focusing on giving effective algorithms for second-order optimization~\cite{hessianfree} or batch gradient methods~\cite{expo}. However, those algorithms have too many constraints about the optimization problem and cannot be applied to most general applications.}

Stochastic Gradient Descent~\cite{sgd} (SGD) is one of the most popular stochastic optimization methods. 
Theoretical results are well studied in~\cite{sgdconvergence}.
However, 
\cite{sgdvr} has shown that the variance in stochastic gradient is 
the key factor limiting the convergence rate of SGD. 
Consequently,
 many SGD variants such as SAG~\cite{sag}, SVRG~\cite{svrg}, S3DG~\cite{s3gd} \eat{, Catalyst~\cite{catalyst}} have been developed to reduce the variance.% in the stochastic gradient. 
 The convergence rate of these variants has been greatly improved 
 in both theory and practice in terms of the number of iterations required to reach a certain accuracy. 
 However, the optimization cost of these methods are not negligible, 
 causing the training cost per iteration to increase substantially.

There are also studies~\cite{adagrad} on the effect
of learning rate on the convergence rate of SGD. 
Naturally, reducing the multiplier of gradient in updates will 
reduce the variance in each update.
This idea motivates us to study if we can scale 
down those stochastic gradients with larger variance by using a smaller learning rate, 
while making up the effects of those gradients by increasing their sampling frequency.
Based on this intuition, 
we propose to accelerate
 the SGD training based on the idea of active learning~\cite{activesurvey1,activesurvey2}. 
 Active learning was originally proposed to select a 
 set of labeled training data to maximize the accuracy of model. 
 \cite{activesampling2} uses the idea of weighted sampling 
 to maximize the information gain of active learning. 
 However, in our Active Sampler, 
 all training data are already labeled, 
 and the active selection is to maximize the learning speed of a passive learning model.

Active Sampler is also related to feature selection methods~\cite{systemfs}. 
Both of them assume that not all the training data are informative for model construction. 
The difference is that feature selection methods find the most informative columns in the training data, 
 whereas Active Sampler finds the most informative rows.

\section{Conclusion}
\label{sec:conclusion}

SGD algorithms are playing a central role in the model training of complex data analytics, 
where sampled training data are used at each training iteration. 
%There are many ways to implement the core sampling operator required.
Uniform sampling and sequential access have been commonly used due to their simplicity. 
In this paper, we study how the sampling method can affect the training speed as a means to facilitate analytics at scale.
Based on the inspiration from active learning, 
we propose Active Sampler which has
sampling frequency that is proportional to the magnitude of gradient.
We show the correctness and optimality of Active Sampler in theory, 
and developed a set of schemes to make the implementation light-weight.
Experiments show that Active Sampler can speedup the training procedure of SVM, feature selection and deep learning by 1.6-2.2x, compared with the uniform sampling. 
Results also demonstrate that Active Sampler has a significant effect on 
reducing the variance of the stochastic gradient, making the training process much more stable.

\small

\bibliographystyle{abbrv}
\bibliography{assgd}

\end{document}